\documentclass[conference]{IEEEtran}
\usepackage{tikz}
\usepackage{amsmath}

\usepackage{filecontents}

\usepackage{subfigure}    
\usepackage{booktabs}
\usepackage{diagbox}
\usepackage{ulem}
\newcommand{\nop}[1]{}
\newcommand{\fig}{Figure}
\newcommand{\aA}{\alpha_A}
\newcommand{\ai}{\alpha_i}
\newcommand{\ah}{\alpha_h}
\newcommand{\ak}{\alpha_k}
\newcommand{\ar}{\gamma}

\newtheorem{theorem}{Theorem}
\newtheorem{proof}{Proof}

\ifCLASSINFOpdf
\else
\fi
\begin{document}
%
\title{Partial Selfish Mining for More Profits}

\author{\IEEEauthorblockN{Jiaping Yu\IEEEauthorrefmark{2}\IEEEauthorrefmark{3},
Shang Gao\IEEEauthorrefmark{3},
Rui Song\IEEEauthorrefmark{3}, 
Zhiping Cai\IEEEauthorrefmark{4} and
Bin Xiao\IEEEauthorrefmark{3}}
\IEEEauthorblockA{\IEEEauthorrefmark{2}College of Computer\\ Ocean University of China\\
email:yujiaping19@ouc.edu.cn}
\IEEEauthorblockA{\IEEEauthorrefmark{3}Department of Computing\\ The Hong Kong Polytechnic University\\
email:\{shanggao, csrsong, csbxiao\}@comp.polyu.edu.hk}
\IEEEauthorblockA{\IEEEauthorrefmark{4}College of Computer\\ National University of Defense Technology\\
email:zpcai@nudt.edu.cn}}

\IEEEoverridecommandlockouts
\makeatletter\def\@IEEEpubidpullup{6.5\baselineskip}\makeatother
\IEEEpubid{\parbox{\columnwidth}{
    Network and Distributed System Security (NDSS) Symposium 2023\\
    28 February - 4 March 2023, San Diego, CA, USA\\
    ISBN 1-891562-83-5\\
    https://dx.doi.org/10.14722/ndss.2023.23xxx\\
    www.ndss-symposium.org
}
\hspace{\columnsep}\makebox[\columnwidth]{}}

\maketitle

\begin{abstract}
Mining attacks aim to gain an unfair share of extra rewards in the blockchain mining. Selfish mining can preserve discovered blocks and strategically release them, wasting honest miners' computing resources and getting higher profits. 
Previous mining attacks either conceal the mined whole blocks (hiding or discarding), or release them completely in a particular time slot (e.g., causing a fork).
In this paper, we extend the mining attack's strategy space to partial block sharing, and propose a new and feasible Partial Selfish Mining (PSM) attack. We show that by releasing partial block data publicly and attracting rational miners to work on attacker's private branch, attackers and these attracted miners can gain an unfair share of mining rewards. We then propose Advanced PSM (A-PSM) attack that can further improve attackers' profits to be no less than the selfish mining. Both theoretical and experimental results show that PSM attackers can be more profitable than selfish miners under a certain range of mining power and network conditions. A-PSM attackers can gain even higher profits than both selfish mining and honest mining with attracted rational miners.
\end{abstract}


%

\section{Introduction}

Decentralized digital currencies like Bitcoin have captured public interest for years~\cite{glaeser2022foundations,wu2020revenue}. By 2023, Bitcoin has a market value of more than 401 billion USD. In Bitcoin, the first miner who solves the puzzle and broadcasts the result will be rewarded with 6.25 bitcoins (BTC), which are worth 130,668.75 USD by January 2023 \cite{capitalcounselor.com}. The more computing resources the miner applies, the more likely it can solve the puzzle and get the bonus first~\cite{wood2014ethereum,nakamoto2008bitcoin}. The result of this puzzle is known as ``Proof-of-Work" (PoW) in Bitcoin. 

However, the high rewards of Bitcoin mining have also made it a valuable target for attackers \cite{zhang2020txspector,lewis2021does,aumayr2021blitz,kiayias2022peredi,kogias2016enhancing}. Previous work has shown that malicious miners can launch attacks by not following the standard mining process, e.g., hiding or discarding mined blocks, releasing a block to cause a fork \cite{torres2021frontrunner,wu2020survive,wang2021securing,tsabary2021mad}. This kind of attack is generally referred to as \textit{mining attacks}. 

\begin{figure}[htbp]
	\centering
	\subfigure[Selfish Mining]{\label{fig:Sel}\includegraphics[width=0.45\columnwidth]{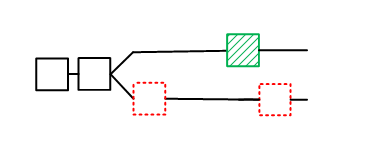}}
	\subfigure[Partial Selfish Mining]{\label{fig:Part}\includegraphics[width=0.45\columnwidth]{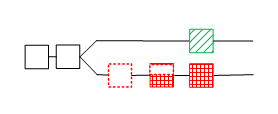}}
	\caption{(a) Selfish mining withholds discovered blocks in its private chain; (b) Partial selfish mining can firstly withhold a discovered block, then share partial block data, and finally broadcast the whole block.}
	\label{fig:Comparison}
\end{figure}


Selfish mining \cite{9eyal2014majority} is one of the most fundamental and well-known mining attacks. In selfish mining, instead of following mining rules and releasing a discovered block immediately, an attacker can withhold the block and continue mining on it as a private branch, as shown in Figure \ref{fig:Sel}. When other miners find a new block, the attacker can immediately cause a fork by releasing the withheld block. If the attacker's fork is selected as the main chain, honest miners' mining power is wasted.


Following the idea of concealing blocks in selfish mining, extensive lines of work have explored variant mining attack tactics against the PoW blockchain. Attackers can either apply these tactics individually or strategically combine them to launch sophisticated mining attacks. For instance, using power splitting, block withholding, and selfish mining, Kwon et al. propose Fork After Withholding (FAW) attacks~\cite{kwon2017selfish}. Besides mining attacks, researchers also proposed other incentive-based DoS attacks in recent years. For example, Michael et al. believe most miners are \textit{rational} in the sense that miners are inclined to choose the most profitable mining strategy to work (i.e., mining on which chain)~\cite{mirkin2020bdos}.



Selfish mining is not reported frequently in existing cryptocurrencies because selfish miners cannot find practical ways to launch the attack easily. It is widely accepted that the attacker's computational power needs to exceed 33\% to gain higher profits than honest mining. Attackers need to either occupy more than 33\% of mining power by themselves or by colluding with other greedy miners to ensure a successful selfish mining attack. 

It is impractical to launch selfish mining by a single miner for large-scale public blockchains. On the one hand, it is hard to occupy more than 33\% of mining power to ensure successful selfish mining, e.g., Bitcoin. According to \cite{blockexplorer.com}, the largest mining pool in Bitcoin only occupies 16\% of overall mining power. On the other hand, it is a general consensus that a large pool is not in the best interest of the blockchain community. Miners are vigilant when a mining pool controls a large fraction of computation power and may unite to boycott the large mining pool \cite{MinerBoycott} and force the pool manager to cut down the mining power. 

The idea of colluding with other rational miners and forming a coalition seems to be an applaudable way to mine selfishly. The overall mining power of the coalition may be enough to launch selfish mining. In practice, it is challenging to build trust among rational miners. Rational miners cannot know whether it is profitable to join the attacker's branch, nor have the assurance of extra profit from attackers.


In this paper, we propose a new block-sharing strategy in mining, called \textbf{partial block sharing} to solve the above challenges. Different from previous new block hiding or revealing, partial block sharing will only reveal part of a block (named \textit{partial block})  while some fields are hiding, e.g., \textit{nonce} and part of arbitrary bytes in the coinbase transaction. Here, we denote the hidden data as \textit{secret}. Previous work~\cite{7eyal2015miner,9eyal2014majority,13bendiksen2018bitcoin} shows that some miners may choose the most profitable mining strategy. In this paper, we define these miners as \textit{rational miners}. Partial block suffices to mine a consecutive second block, with all transaction outputs (related to UTXO) and block header hash value available. A partial block is regarded as an invalid block by honest miners but does not hinder rational miners from mining after it if it is more profitable. 

Based on the partial block sharing strategy, we propose a new and practical mining attack called \textbf{Partial Selfish Mining} (\textbf{PSM}). As shown in Figure \ref{fig:Part}, PSM starts as selfish mining to withhold a newly mined block. Then, the attacker can launch the partial block-sharing strategy and finally releases the secret by broadcasting the whole block before or right after 
another miner finds a new block. 
A rational miner can get the partial block before a new block is found by other miners on the public chain and may be attracted to work on the attacker's private branch, which is named as \textit{attracted} miner in this paper. When more attracted miners join the private branch, it has a higher probability of being the main chain. Thus, compared with selfish mining, PSM can not only waste the mining power of miners spending on the original public chain but also greatly enhance the success rate of the attacker's private branch to be chosen as the longest main chain. While previous mining attacks cannot improve selfish mining itself, PSM provides \emph{a new colluding strategy} to allow the attacker to gather more mining power in attacker's branch significantly in a short period. For instance, if the attacker's mining power is only 20\% of the whole system, PSM can attract profit-driven rational miners to attacker's branch and the overall mining power contributing to the selfish mining branch will be more than 20\%, and possibly even larger than 1/3.
Since estimating the exact fraction of rational miners in the wild is challenging, we analyze the attacker's revenue with different fractions of rational miners.
 Analysis results show that PSM can have a higher reward than both honest and selfish mining in a certain range. It is always beneficial for a rational miner to join the attacker's private branch if his mining power is smaller than the attacker's. 

To make PSM practical and colluding strategy successful, we must have mechanisms to convince rational miners that it is profitable to mine in the attacker's private branch. First, the attacker needs to ensure that it indeed has the secret, i.e., a complete and valid new block. Otherwise, all the mining power that rational miners spend on the private branch is in vain. Second, the attacker may deny broadcasting its secret as promised. We name this sabotaging behavior as PSM Denial of Service (PSM-DoS) attack. To address the first concern about block validation, we propose a zero-knowledge-proof-based mechanism to prove the block possession, which is entirely owned by the attacker. To counteract PSM-DoS attacks, we design an economic-based profit protection mechanism for rational miners. If the attacker fails to broadcast the promised secret/block timely, his deposit in a trusted third party could be forfeited by a rational miner. This mechanism also ensures that the secret can be calculated by others with limited computing power in an acceptable period of time, and makes the attacker's PSM-DoS attack economically not worthwhile. We implement the profit protection mechanism with a smart contract.


Moreover, we propose an \textbf{Advanced PSM} (\textbf{A-PSM}) strategy for attackers to gain more profits when the attacker can get the block height of the public branch timely. An attacker has economic incentives to monitor the  block height of the public branch to apply A-PSM actively.
This block height can also be reported by attracted miners or obtained via decentralized oracle~\cite{zhang2020deco} in case of any dispute to guarantee attracted miners' rewards. 
We show that the attacker's profit is no less than selfish mining in A-PSM. In most cases, attracted miners will follow the attacker's branch to have more mining rewards. Thus, we believe that A-PSM is a lucrative alternative to selfish mining in mining-related attacks.


We analyze the scenario when multiple attackers adopt either PSM or A-PSM. Attracted miners should join the branch that publishes the partial block  early to get more mining time and expected rewards. If attackers release partial blocks simultaneously in distinct branches, joining any branch has no impact on attracted miners' profits in PSM. However, they should join the branch with a larger attacker mining power in A-PSM.




Our theoretical analysis and evaluations show consistent results that  the PSM attack can increase the attacker's profits in a large fraction of the parameter space, and A-PSM attack can always earn more profits than selfish mining. In PSM, when 50\% of miners are rational, an attacker with 20\% mining power can get 1.25\% and 9.79\% more profits than honest mining and selfish mining, respectively. In A-PSM, when 30\% of miners are rational, an attacker with 10\% mining power can have a higher reward, gaining at most 23.6\% and 13.1\% more profits than the honest mining (i.e., behaving as a normal miner without attacks) and selfish mining, respectively. 

In summary, we have made the following contributions: 
\begin{itemize}
    \item We develop a new and practical block-sharing strategy called partial block sharing. Based on it, we propose a new mining attack protocol, PSM. PSM exhibits a new paradigm of colluding with other miners to gain mining advantages. By sharing the partial block data and attracting rational miners to work on the attacker's private branch, the attacker and attracted miners can gain more mining rewards than selfish mining under certain conditions.
    \item To make PSM feasible and prevent potential PSM-DoS attacks, we propose two mechanisms to convince rational miners to join the attacker's private branch. The first zero-knowledge-proof-based mechanism can prove the new whole block possession owned by the attacker. The second secret computation mechanism can make PSM-DoS attacks not worthwhile for attackers.
    \item To further increase the attacker's profits, we proposed A-PSM at the cost of acquiring the block height of the public branch timely. A-PSM assures an attacker of revenue no less than selfish mining. It can even outperform honest mining when the overall mining power is no more than 50\% in the attacker's private branch (including attracted miners).
    \item
    We analyze the profit and show the \textit{optimal strategy} for attackers and rational miners to maximize their revenue by adopting PSM and A-PSM strategies respectively in various network settings.
    \item We also discuss the scenario of multiple attackers. We evaluate and compare the revenue of both attackers and rational miners in the context of PSM and A-PSM strategies with the one of selfish mining and honest mining. In PSM, attracted miners can have more profits when their mining power is less than the attacker. An attacker adopting PSM can be more profitable than adopting selfish mining under a certain range of mining power and network conditions. With a realistic fraction of rational miners, A-PSM can always earn a higher profit than both selfish mining and honest mining.
\end{itemize}

In Section~\ref{sec:Preliminaries}, we present an overview of the blockchain background and related work. Section~\ref{sec:ModelAndAssump} describes the mining model and assumptions. Then we describe the PSM attack in Section~\ref{sec:ATKOverviewAndStr} in detail, including the profit analysis of rational miners. 
In Section~\ref{sec:simulation}, we use numeric analysis to evaluate the reward of PSM strategy by comparing it with both honest and selfish mining. Then we propose the A-PSM strategy in Section~\ref{sec:APSM}. In Section~\ref{sec:Eval}, we further illustrate the optimal mining strategy for attackers and rational miners, respectively. In Section~\ref{sec:Discussion}, we discuss the multiple-attacker scenario and countermeasures against PSM and A-PSM strategies. Section~\ref{sec:Conclusions} concludes the paper.

\section{Related Work}\label{sec:Preliminaries}

\textbf{Selfish Mining:} Selfish mining was first proposed by Eyal et al.~\cite{9eyal2014majority}. A selfish mining attacker can earn extra rewards by intentionally generating a fork. When an attacker discovers a new block in selfish mining, it will keep the block as its private branch and keeps mining after it. When other miners find a block, the attacker will release the withheld blocks to cause a fork. The attacker can earn extra rewards once its private branch is selected as the main chain. According to~\cite{9eyal2014majority}, miners with more than $33\%$ computational power can surely get an extra reward compared with honest mining. In \cite{nayak2016stubborn}, researchers find that for a larger parameter space, by following a more `stubborn' strategy, miners can gain more rewards than selfish mining when their mining power is large enough.  Liao et al. \cite{liao2017incentivizing} present the whale transactions. By deploying large fees, attackers can incentivize miners to fork the blockchain. In recent years, Negy et al. \cite{negy2020selfish} further analyzed the profits of selfish mining and proposed intermittent selfish mining. It assures the attacker can get more profits than honest mining even after the difficult adjustment. Li et al.~\cite{li2021semi} analyze the mining attack strategy from the honest miner's perspective and optimize the selfish mining with a hidden Markov decision process. Sapirshtein et al.\cite{sapirshtein2016optimal} propose an extension to the model of selfish mining attacks in the Bitcoin protocol, providing an algorithm to find optimal policies for attackers and tight upper bounds on revenue. 
Our methods can further improve the attacker's rewards by attracting other miners to work in the attacker's branch.
For instance, consider a scenario where half of the miners are rational with parameters $\aA=1/3$ and $\ar=0$. In this case, the revenue of the original selfish mining strategy is 1/3. The strategy proposed in \cite{sapirshtein2016optimal} achieves a revenue of 0.33707, which is 1.12\% higher than the original selfish mining strategy. However, with Partially Selfish Mining (PSM), the revenue attained is 0.3403, which is 0.96\% higher than the revenue achieved by the strategy in \cite{sapirshtein2016optimal} and 2.1\% higher than the revenue of the original selfish mining strategy.

\nop{
\begin{table}[]
\resizebox{\columnwidth}{!}{
\begin{tabular}{cc}
\hline
Name                    & Description                                                                                                                                              \\ \hline
Selfish Mining\cite{9eyal2014majority}          & \begin{tabular}[c]{@{}c@{}}Withholding the discovered block and continue\\  finding the next block on its private branch.\end{tabular}                   \\ \hline
Block Withholding (BWH)\cite{luu2015power} & \begin{tabular}[c]{@{}c@{}}Discarding the block found in the victim pool\\  to reduce the victim pool's revenue.\end{tabular}                             \\ \hline
Bribery\cite{bonneau2016buy}          & \begin{tabular}[c]{@{}c@{}}Pay other miners with bribes when\\  extending its branch in racing.\end{tabular}                                             \\ \hline
Power Splitting\cite{kwon2017selfish}          & \begin{tabular}[c]{@{}c@{}}Strategically splitting the mining power and using part\\  of powers to sabotage the reward of victim pool.\end{tabular} \\ \hline
Power Adjusting\cite{gao2019power}         & \begin{tabular}[c]{@{}c@{}}Adjusting the mining power between the victim pool\\  and solo mining to get a higher reward.\end{tabular}           \\ \hline
\end{tabular}}
\caption{Typical strategies for mining attacks.}
\label{tab:MiningAttack}
\end{table}
}
\textbf{Other Cryptocurrency-Related Attacks:} After the proposal of selfish mining, researchers have proposed various mining attack tactics against the PoW Blockchain. 
For those who do not have enough mining power, Loi et al.~\cite{luu2015power} proposed the block withholding (BWH) attack. 
By strategically splitting the mining power, attackers can get more rewards than honest mining in the long run. Kwon et al.~\cite{kwon2017selfish} combine the BWH with selfish mining and propose the Fork After Withholding (FAW) attack. Instead of simply discarding the full block finding in the victim pools, in FAW, attackers mine on the newly mined block to generate a private branch. 
In recent years, Gao et al.~\cite{gao2019power} proposed a Power Adjusting Withholding (PAW) attack. PAW attackers adjust the mining power between the victim pools and solo mining, unlike FAW attackers. By doing so, PAW attackers can gain twice as much revenue as FAW attacks.


Besides mining attacks, researchers also proposed other incentive-based attacks. For example, based on the assumption that miners are inclined to choose the most profitable
mining strategy to work, Michael et al.~\cite{mirkin2020bdos} propose the Blockchain Denial of Service (BDoS) attack. Instead of getting a higher reward, BDoS attacker invests resources to reduce other miners' profits and lure them away from mining. By publishing the block header of the newly mined block, the attacker signal to the miners that the system is in a state that reduces its profits. 
In our study, PSM attackers share the full block except for the secret data instead of sharing the block with the rational miner. Miners receiving partial block data are not likely to suffer a loss. Instead, they can reevaluate their profits and choose the most profitable way to work. 

When there is a fork and a race occurs, miners follow the block they received first as the legal block. In previous work ~\cite{neudecker2019short}, researchers tend to believe that each branch has about 50\% of miners to work on. 
However, Saad et al. investigated the Bitcoin safety properties and concluded from different angles. In~\cite{saad2021revisiting}, they pointed out that the assumption about the strong network synchrony may not hold in real-world deployment. They also implemented a HashSplit attack that allowed the attacker to orchestrate the mining power distribution when a race occurred. Their work~\cite{saad2021syncattack} further shows that the unstable hash rate distribution can make it possible for attackers to launch a double-spending attack without mining power.

\section{Model and Assumptions}\label{sec:ModelAndAssump}
\subsection{Mining Model}

We consider a blockchain system with $n$ miners whose normalized mining powers are denoted by $\alpha_1, \alpha_2 ... \alpha_n$, and an attacker with mining power $\alpha_A$, such that $\alpha_A+\sum_{i=1}^n\alpha_i=1$. 



As shown in Figure~\ref{fig:Miners}, in PSM attacks, without loss of generality, we assume that miners are divided into the following groups: 
\begin{figure}[htbp]
	\centerline{\includegraphics[width=0.8\columnwidth]{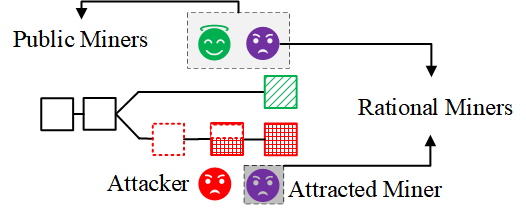}}
	\caption{\textbf{Miners' roles in PSM scenario.}} 
	\label{fig:Miners}
\end{figure}

\begin{itemize}
    \item \textbf{Attacker:}  A miner or a colluding minority pool that has a newly mined block(s) and follows the PSM strategy. It can preserve a mined block(s), form a private branch, and share partial block data with rational miners like~\cite{mirkin2020bdos}. Besides, it also has access to the smart contract to share the proof of block possession and other needed data with rational miners.
    \item \textbf{Rational miners:}  A group of miners that are profit-driven but still follow PoW to mine. In most cases, this group is the minority.
    Receiving a partial block released by an attacker as in PSM, rational miners can choose their optimal strategy (i.e., mining on which branch) to get a higher reward. 
    \item \textbf{Attracted miners:}  Part of rational miners that choose to work on the attacker's private branch, mining after the partial block. We name the strategy taken by attracted miners as the \textit{greedy mining strategy}.
    \item \textbf{Public miners:} Miners that take honest mining to truthfully generate and release new blocks in the longest chain. We name the strategy taken by public miners as the \textit{public mining strategy}. Non-attracted rational miners are considered public miners and follow the public mining strategy.
\end{itemize}

Note that all attracted miners are rational miners. The reverse does not always hold since some rational miners may choose to mine on the original public branch according to their optimal strategy even under PSM attacks. A miner, rather than the attacker, is either an attracted miner or a public miner under attacks. If another miner (e.g., a non-public miner) finds a new block and launches the PSM attack simultaneously, we have multiple attackers and such cases are discussed in Section \ref{sec:Discussion}.


We denote the rushing ability of the attacker by $\gamma$. If the attacker publishes a new whole block from its private branch to race with other miners (e.g., a new block broadcast in the public branch), $\gamma$ is the expected ratio of public miners that receive the attacker's block first. $\gamma$ mainly depends on the miner's network condition. 
In previous work, $\gamma$ is commonly considered as $\frac{1}{2}$~\cite{9eyal2014majority,sapirshtein2016optimal}. However, some researchers also propose methods that allow the attacker to get a larger $\gamma$~\cite{saad2021revisiting}. In this paper, we consider $0<\gamma<1$.
\subsection{Assumptions}
We have the following assumptions that are consistent with other PoW-based mining attacks, such as \cite{9eyal2014majority}, \cite{nayak2016stubborn}, \cite{gervais2016security}, \cite{gao2019power}.

\begin{enumerate}
    \item Instead of 6.25 BTC, the reward of a new block is normalized to 1. Our analysis gives the expected reward for every participant \cite{13bendiksen2018bitcoin}.
    \item Each miner/pool's computational power should not be greater than 0.5 to avoid ``51\% attacks" \cite{bradbury2013problem}. 
    \item A miner's expected normalized reward equals the probability of finding a valid block in each round. This assumption is reasonable because according to \cite{gervais2016security}, the probability of unintentional forks is around 0.41\%, which is negligible.
    \item Rational miners do not trust the attacker and join its private branch unless the attacker can provide an extra method to guarantee the profits of rational miners. Consequently, the traditional selfish miner cannot attract rational miners.
    \item Not all the miners in the blockchain system are rational. Though we proposed some measures to address the concern of miners working on the attacker's private branch, we take the possibility into consideration that they cannot allay the concerns of all miners. Thus, we discuss the fraction of rational miners to be no more than 50\% in this paper.
    \item Rational miners tend to pull external information and adjust their mining strategy accordingly. This assumption is commonly made in previous work~\cite{mirkin2020bdos,eghbali201912}.
\end{enumerate}

\section{Partial Selfish Mining Attack}\label{sec:ATKOverviewAndStr}

In PSM, an attacker follows the partial block sharing strategy and shares the partial block information with rational miners. The partial block has some data covered, e.g., \textit{nonce} and part of arbitrary bytes in the coinbase transaction. Miners can mine after it to get a new block. The hidden data can be recovered by others, spending considerable mining power. We denote the hidden data as \textit{secret}. 

The generic PSM can first withhold the new block, then release the partial block, and finally broadcast the complete block. If an attacker does not release the partial block or complete block in PSM, it becomes selfish mining. The attacker can also bypass the first two steps and directly publish the complete block, and it is honest mining. Or the attacker can bypass the first step (when selfish mining is not the dominant strategy as discussed in Section \ref{sec:Eval}) but continue the second and third steps, which has not been discussed previously and is the focus of this section.

\subsection{Attack Overview}


\begin{figure}[htbp]
	\centerline{\includegraphics[width=\columnwidth]{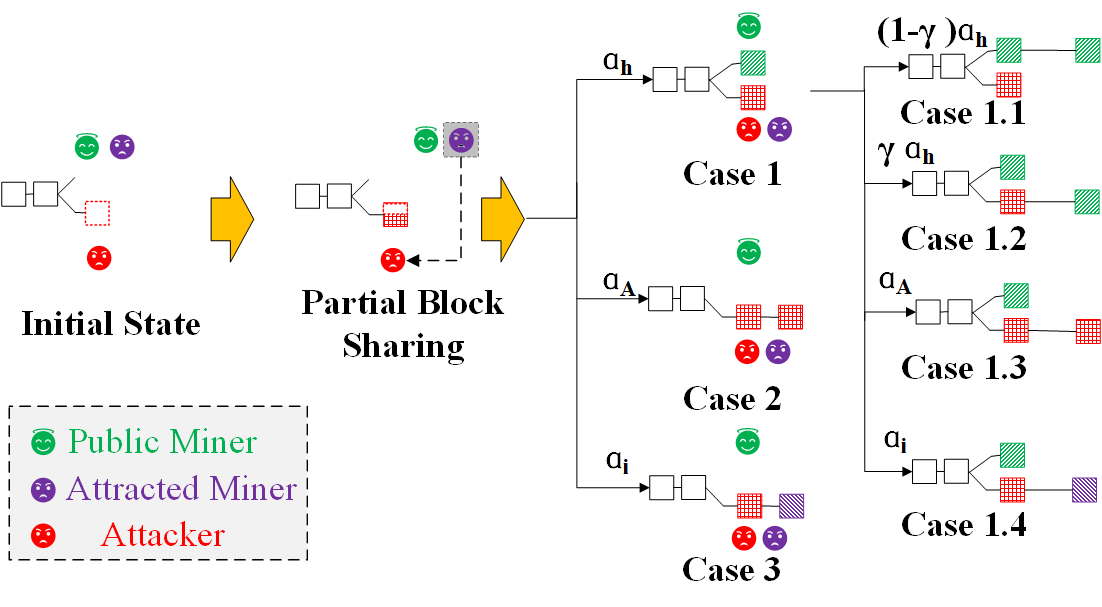}}
	\caption{\textbf{Workflow of PSM strategy.} Instead of publishing a new block, the attacker shares partial block data with other miners to attract them to join its private branch.Three possible cases of finding another new block after the announcement of the partial block. Case 1: By public miners; Case 2: By the attacker; Case 3: By attracted miners.}
	\label{fig:SMOverview}
\end{figure}

The workflow of PSM is shown in Figure \ref{fig:SMOverview}. In PSM, when the attacker finds a block, it keeps the block private instead of immediately releasing it. In the meantime, the attacker releases the partial block with proof of block possession. With these released data, rational miners could mine on the private branch. To assure the profits of attracted miners, the attacker will also announce a smart contract that allows the rational miner to get the attacker's collateral if it does not release the secret promised. For simplicity and fairness, we assume that all honest and rational miners will work on the branch first received if a race happens. An attacker can even lure miners to work on his branch in the race case by  adding transactions with lower fees to its block and leaving high-fee transactions to be collected by others. 



After the release of the partial block, three possible cases may happen: the attacker finds the new block on its private branch, public miners find a new block on the public branch or attracted miners find a new block on the private branch. 

In the first case, when miners find a new block on the public branch, the attacker releases its private branch and starts a 0-lead racing. In this scenario, the attacker and the attracted miners will mine on the previously private branch, and public miners will choose to mine on either branch. As defined earlier, $\gamma$ public miners will work on the private branch, and $1-\gamma$ public miners will mine on the public branch.

In the second case, when the attacker finds a new block on the private branch, the attacker will release the whole private branch and get the two blocks' revenue. Then the system goes back to the single-branch state.

In the third case, when the attracted miners find a new block in the private branch, the attacker will immediately release the whole private branch to get the revenue of all the blocks on the private branch. Similar to the first case, the blockchain returns to the single-branch state.

Attackers can decide among PSM, honest or selfish mining based on profitability. In the following, we illustrate the reward of PSM. The reward of honest and selfish mining is given in Section~\ref{sec:simulation} for comparisons. We show the optimal mining strategy for the attacker in Section~\ref{sec:Eval}.

\subsection{PSM Reward}

We use the following parameters in our analysis. 
\begin{itemize}
    \item$\gamma$: Ratio of  public miners choosing attacker's branch when a race occurs. 
    \item $\alpha_i$: Mining power of attracted miners.
    \item $\alpha_k$: Mining power of one rational miner $k$.
    \item $\alpha_A$: Mining power of the attacker.
    \item $\alpha_h$: Mining power of public miners.
    \item $R^n_m$: Revenue of miner $n$ in case $m$.
    \item $H$: Public mining strategy.
    \item $S$: Selfish mining strategy.
    \item $P$: Partial selfish mining strategy.

\end{itemize}

\begin{figure}[htbp]
	\centerline{\includegraphics[width=0.4\columnwidth]{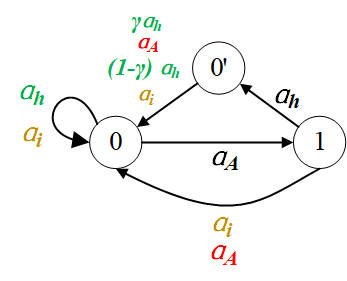}}
	\caption{\textbf{State machine of PSM when miners with overall $\alpha_i$ mining power working greedy.}} 
 
	\label{fig:STPSMATK}
\end{figure}

\fig~\ref{fig:STPSMATK} shows the progress of the PSM attack as a state machine. The state of the system shows the lead status of the attacker's private branch. Zero lead is separated into states 0 and 0'. State 0 means there is no branch in the blockchain, and the public chain is the longest chain. State 0' is the state where there are two branches on the blockchain, one is the original public chain, and the other is the attacker's previous private chain. The value on each arrow indicates the possibility of state transition. 

Based on the above state machine, we can derive the PSM attacker's profit as follows:

\begin{theorem}
The profit of a PSM attacker is
\begin{small}
\begin{equation}
\begin{aligned}
\label{equ:atkprof}
   R^A_P=\frac{\aA\ah^2\ar+(\aA\ah+\aA)\ai+2\aA^2\ah+2\aA^2}{(2\aA\ah+2\aA+1)\ai+2\aA\ah^2+(2\aA^2+1)\ah+2\aA^2}
    \end{aligned}
\end{equation}
\label{thm:PSMATK}
\end{small}
\end{theorem}

The proof of Theorem~\ref{thm:PSMATK} is given in Appendix~\ref{sec:PSMAtk}.

For simplicity, we exclude the block reward variable and network-related factors other than rushing ability in our analysis. This is due to the relatively minor impact of these factors when compared with the mining power and rushing ability.

\subsection{Rational Miners' Profits in PSM}\label{sec:MinerReward}

To attract rational miners to work on its private branch, the attacker needs to assure that the rational miners can get more rewards than public mining, and make a solid promise to guarantee that the attacker's cost of breaking the promise is unbearably high. In this section, based on our assumption in Section \ref{sec:ModelAndAssump}, we analyze the rational miner's strategy space and its profits. 

\begin{figure}[htbp]
	\centerline{\includegraphics[width=0.8\columnwidth]{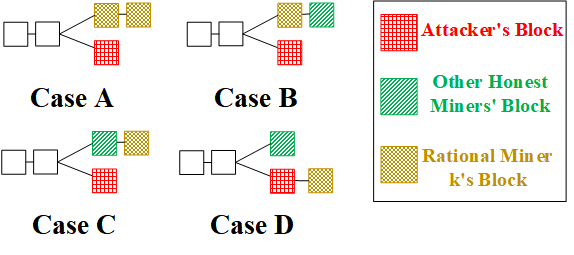}}
	\caption{\textbf{Four reward cases exist when a rational miner \textit{k} follows the public mining strategy.} Case A: Miner \textit{k} finds two blocks on the public branch. Case B: Miner \textit{k} finds one block on the public branch, and a public miner further extends the block on the public branch. Case C: Another public miner finds a new block on the public branch, and miner \textit{k} further extends the block on the public branch. Case D: Another public miner finds a new block on the public branch, then miner \textit{k} finds a new block on the attacker's branch. Please note that once a public miner finds a block, the attacker will surely release the full block and start a race.} 
	\label{fig:RatHonest}
\end{figure}

Assuming in a bitcoin-like blockchain network, rational miner $k$ will choose the most profitable branch to work based on the information it received. Miner $k$ can know the $\alpha_A$ from the attacker's message. In a balanced blockchain network with low congestion, $\gamma$ is usually $\frac{1}{2}$~\cite{9eyal2014majority,sapirshtein2016optimal}. In practice, the network status changes dynamically, and miners can hardly get the exact $\gamma$ value. Rational miners could estimate the $\gamma$ value based on the current network status or simply consider the worst case ($\gamma=0$). We are interested in which strategies can maximize the miner's profit with different mining power, $\gamma$, and $\alpha_A$.

From the state machine in \fig~\ref{fig:STPSMATK}, The miner $k$'s revenue $R_G^k$ can be expressed as follows: 
\begin{theorem}
When miner $k$ follows the greedy mining strategy, miner $k$'s revenue is 

\begin{small}
\begin{equation}
R_G^k=\frac{(\aA\ah+\aA+1)\ak}{(2\aA\ah+2\aA+1)\ak+2\aA\ah^2+(2\aA^2+1)\ah+2\aA^2}. 
\label{eq:RGK}
\end{equation}
\end{small}
\label{thm:PSMGMN}
\end{theorem}
\begin{proof}
    According to Appendix~\ref{sec:PSMAtk}, the overall revenue of all attracted miners can be expressed as: 
    \begin{small}
    \begin{equation}
\begin{aligned}
   R^G_P=\frac{(\aA\ah+\aA+1)\ai}{(2\aA\ah+2\aA+1)\ai+2\aA\ah^2+(2\aA^2+1)\ah+2\aA^2}
    \end{aligned}
    \label{eq:RGP}
\end{equation}
    \end{small}

  Rational miners make their decisions independently. Therefore, miner $k$ faces difficulty in accurately determining the number of rational miners in the network. 
   As a cautious rational miner, miner $k$ should consider the worst-case scenario when the attacker has small mining power, assuming it is the only attracted miner in the network, which means $\ai=\ak$. Then, we have Equation (\ref{eq:RGK}).
\end{proof}

When miner \textit{k} follows the public mining strategy, there are four possible reward cases, as shown in Figure~\ref{fig:RatHonest}.  Its revenue can be derived as follows:
\begin{theorem}
    When miner $k$ follows the public mining strategy, its revenue is:
    \begin{small}
\begin{equation}
\begin{aligned}
    R_P^k=\frac{(\aA\ak^2+(\aA^2-\aA)\ak)\ar+((-2\aA^2)+2\aA+1)\ak}{\aA+1}.
\label{equ:RevMiner}
\end{aligned}
\end{equation}
\end{small}
\label{thm:PSMHMN}
\end{theorem}

The proof of Theorem~\ref{thm:PSMHMN} is in Appendix~\ref{sec:PSMMiner}.


Here, we first consider a specific case: an attacker with the computational power of 0.1 who executes the PSM attack against the blockchain. For different $\alpha_k$, the expected profits of following greedy and public mining strategies are shown in Figure \ref{fig:MinerRev}.

\begin{figure}[htbp]
	\centerline{\includegraphics[width=0.6\columnwidth]{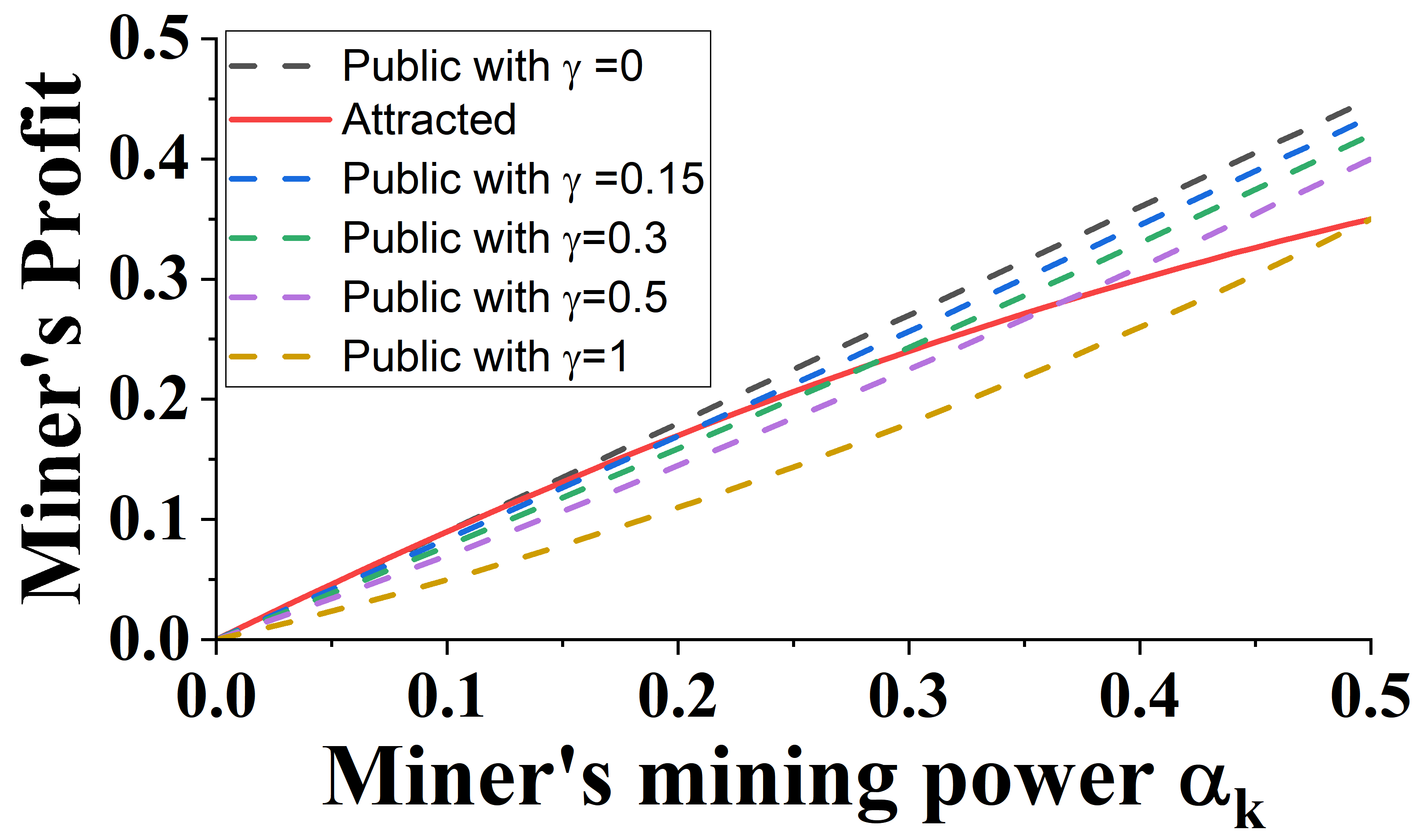}}
	\caption{\textbf{Revenue comparisons for a miner to choose greedy or public mining strategies when attacker's mining power $\alpha_A=0.1$.} Note that the rushing ability $\gamma$ has no impact on the revenue of the greedy strategy. }
	\label{fig:MinerRev}
\end{figure}


To evaluate the profit of the greedy strategy over the public mining strategy, we calculate the relative extra rewards (RER). 
The expected RER of miner $k$ when adopting strategy $s_1$ rather than $s_2$ can be expressed as: 
\begin{small}
\begin{equation}
  RER_{k}^{s_1,s_2}=\frac{R_{s_1}^k-R_{s_2}^k}{R_{s_2}^k},
\end{equation}
\end{small}
where $R^k_{s_1}$ and  $R^k_{s_2}$ represent the reward of $k$ when adopting $s_1$ and $s_2$  strategies respectively.

\begin{theorem}
For rational miner $k$, working on the attacker's private branch is always more profitable when $\alpha_k<\alpha_A$.
\end{theorem}

\begin{proof} Rational miners, including mining pools, have independent mining decisions and processes. It is hard to predict the ratio of attracted miners in the network. When judging his profit, miner \textit{k} should consider the worst case that there are no other attracted miners in the whole network, i.e., $\alpha_i=0$.

The RER of miner $k$ when following greedy mining instead of public mining strategy is:

\begin{small}
\begin{equation}
\begin{aligned}
    RER_k^{G,H}&=\frac{R_k^G-R_k^H}{R_k^H} =\frac{((-\aA\ai)-\aA^2+\aA)\ar-\aA\ai+\aA^2}{(\aA\ai+\aA^2-\aA)\ar-2\aA^2+2\aA+1}.
    \label{equ:PSMMinerRER}
\end{aligned}
\end{equation}
\end{small}

$RER^{G,H}_{k}>0$ means that being greedy is more profitable for the miners. A negative value means that the miner suffers from a loss when choosing the greedy strategy. 

With $\alpha_h=1-\alpha_A-\alpha_k$ ($\alpha_i=0$), when $RER_k^{G,H}>0$, we have:

\begin{small}
\begin{equation}
\begin{aligned}
\label{equ:RERMiner}
    ((-\aA\ak)&-\aA^2+\aA)\ar-\aA\ak+\aA^2>0\\
    \alpha_k&<\frac{\aA-(\aA-1)\ar}{1+\gamma}
    \end{aligned}
\end{equation}
\end{small}

From Equation (\ref{equ:RERMiner}), we can see that the miner's RER mainly depends on $\gamma$ and $\alpha_A$. 
In the worst case when $\gamma =0$, miner $k$ will get more rewards on the attacker's private branch if $\alpha_k<\alpha_A$. Figure \ref{fig:RERMinerQ} shows the quantitative simulation results.
\end{proof}
\begin{figure}[ht]
	\centering
	\subfigure[Rushing ability $\gamma$=0]{\label{minerr0}\includegraphics[width=0.45\columnwidth]{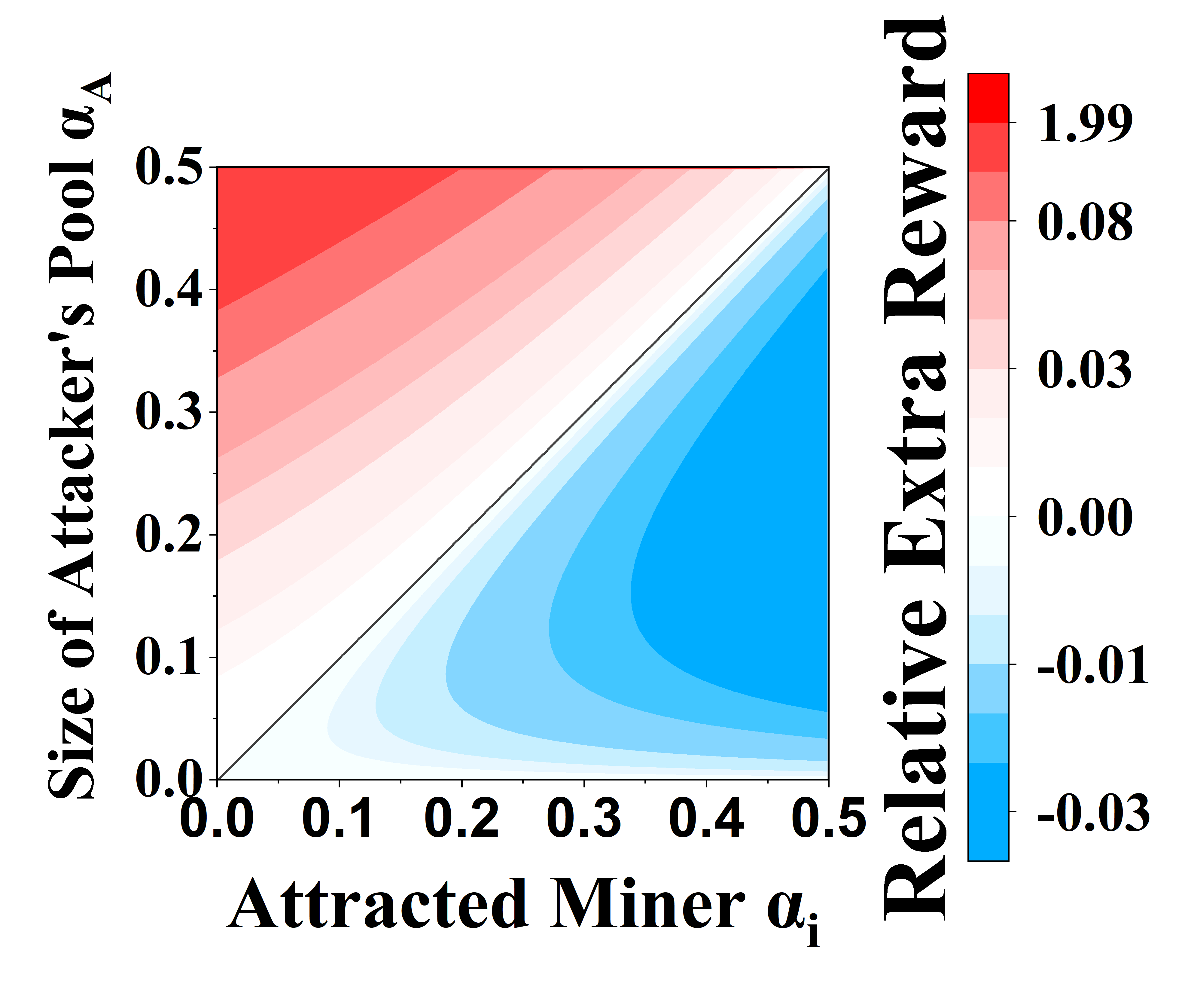}}
	\subfigure[Rushing ability $\gamma$=1]{\label{minerr1}\includegraphics[width=0.45\columnwidth]{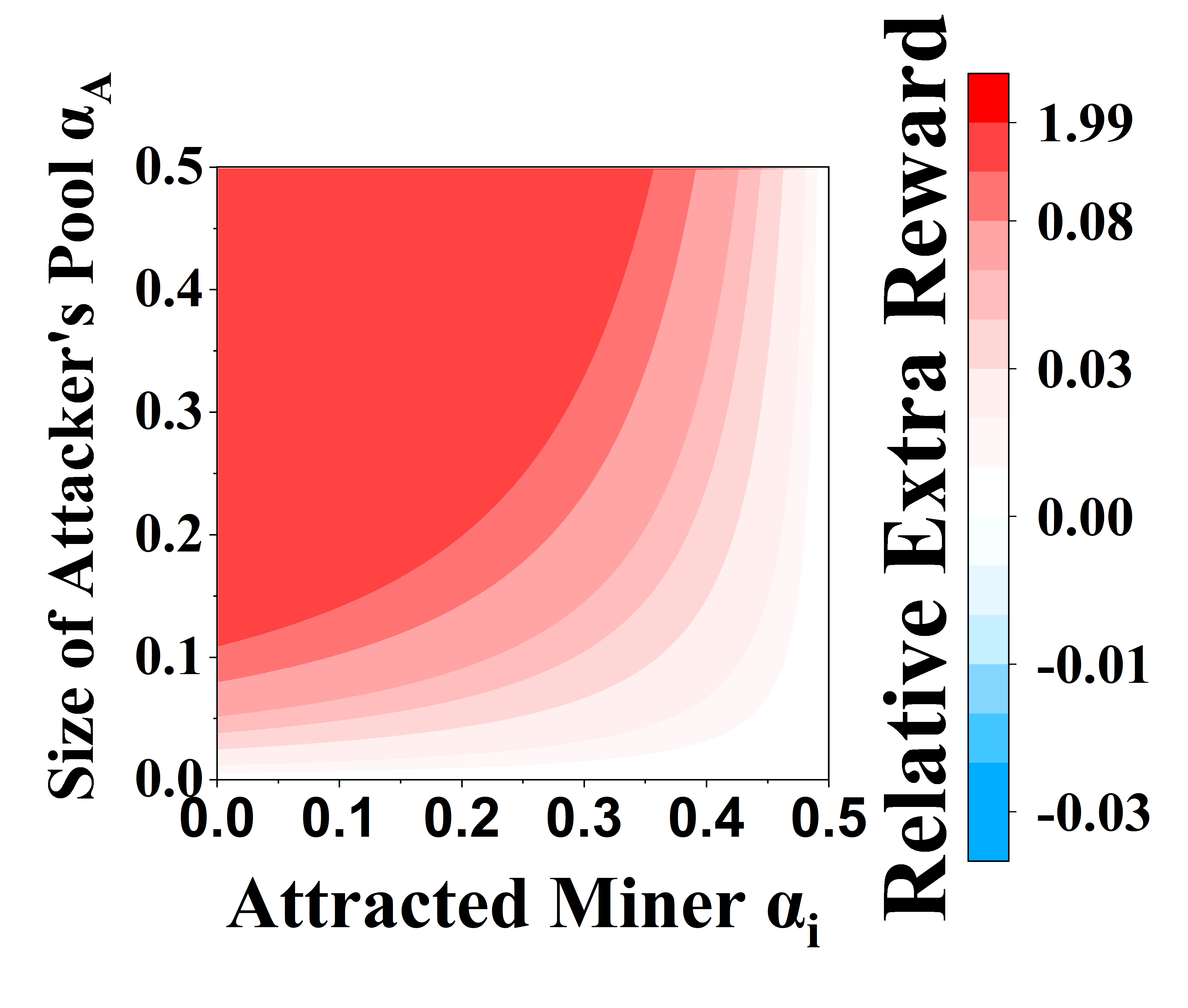}}
	\caption{\textbf{With a PSM attacker, the rational miner's relative extra reward when choosing greedy mining strategy instead of public mining strategy} ($RER^{G,H}_{k}$). The solid line represents no extra reward. For miner $k$, when $\gamma=0$, working on the private branch is more profitable only when $\alpha_k<\alpha_A$.}
	\label{fig:RERMinerQ}
\end{figure}

To verify the theoretical results, we simulate the RER of an attracted miner with a mining power of 0.2, using a Monte Carlo method over $10^9$ rounds, with an upper bound of $10^{-4}$ error. Table \ref{tab:minerRER} shows the Monte Carlo simulation results. The results are consistent with our expectation from Equation (\ref{equ:RERMiner}).

Section~\ref{sec:PracticalConcerns} provides a detailed mechanism design to convince rational miners that working on the attacker's private branch is a profitable endeavor. It includes a thorough examination of various aspects, such as the proof of block possession and a profit protection mechanism that ensures that rational miners receive the revenue promised.
\begin{table}[]
\caption{\textbf{Monte Carlo simulation result of rational miner $k$'s relative extra reward }($RER^{G,H}_{k}$).} 
\label{tab:minerRER}
\centering
\resizebox{\columnwidth}{!}{
\begin{tabular}{@{}cllll@{}}
\toprule
\multicolumn{1}{l}{\diagbox{$\gamma$}{$\alpha_A$}} & \multicolumn{1}{c}{0.1} & \multicolumn{1}{c}{0.2} & \multicolumn{1}{c}{0.3} & \multicolumn{1}{c}{0.4} \\ \midrule
0                    & -5.00(-5.00) & 0.02(0)        & 6.25(6.25)     & 14.30(14.29)   \\
0.25                    & 7.03(7.04) & 12.48(12.50)        & 19.33(19.3)     & 27.95(28.0)   \\
0.5                  & 22.56(22.56) & 28.47(28.47) & 35.96(36.00)   & 45.34(45.45)   \\
0.75                    & 43.46(43.4) & 50.04(50.0)        & 58.19(58.144)     & 68.36(68.42)   \\
1                    & 72.80(72.73) & 79.98(80)     & 88.79(88.89) & 100.02(100.0) \\ \bottomrule
\end{tabular}}
\end{table}

\section{Feasibility of PSM}\label{sec:PracticalConcerns}


In this section, we describe the mechanism to prove that rational miners working on the attacker's private branch are profitable in detail. To assure that the attacker will follow the rules announced, the attacker needs to provide proof of block possession and a profit protection mechanism that assures the rational miners can get the promised revenue. 

\subsection{Proof of Block Possession}
\label{sec:BlockPub}



Specifically, to provide proof of block possession, the attacker writes some random information $r$ to the coinbase transaction of the new block to provide sufficient randomness and uses $r$ and the $nonce$ of the block as the witness. Other parts of the block header besides $nonce$ and $r$, denoted as $b$, can be used as the public statement, including information such as the Merkle root of transactions in the block, and the hash $h$ of the block which satisfies the difficulty requirement. After this, the attacker needs to prove that:

\begin{enumerate}
    \item It knows $nonce$ and $r$ such that $H(b,nonce,r)=h$; and
    \item  The $nonce$ and $b$ are well-formed, i.e., $nonce$ is a 32-bit integer.
\end{enumerate}
To this end, the attacker needs to generate a proof of block possession $\pi_b\overset{\rm R}\leftarrow\mathsf{Prove}((b, h), (nonce,r))$ to prove the following relation
\begin{small}
\begin{equation}
    H(b,nonce,r)=h\ \land\ 0\leq nonce\leq2^{32}-1
\end{equation}
\end{small}
\noindent is satisfied. The attacker then makes the tuple $(\pi_b,b,h)$ publicly available on a dedicated website. In this way, other miners can calculate
\begin{small}
\begin{equation}
    \mathsf{Verify}((b, h),\pi_b)\overset{?} =1
\end{equation}
\end{small}
to verify whether the proof $\pi_b$ holds. If the above verification passes, the other miners can be sure that the attacker holds a specific $nonce$ and $r$, which enables the hash $h$ of the new block to satisfy the difficulty requirement.

If the attacker does not want to share the block information with every miner, it can also share the partial block data $b$ with a specific miner. For this purpose, it can take advantage of the zero-knowledge contingent payment (ZKCP) protocol \cite{campanelli2017zero}. We propose an example block-sharing strategy in Appendix~\ref{sec:BlockSharing}.

To evaluate the performance of proof generation and verification, a computer running Ubuntu 22.04 with a 3.50 GHz Intel i9-11900k CPU and 32 GB of RAM was used. We constructed the above zero-knowledge proofs based on libsnark and implemented the proof of block possession and ZKCP-based block information exchange using about 3100 lines of C++ code.

For the performance of Block System Possession, the proving time is 2.77 s, the verification time is 19 ms, and the proof size is 517 Byte. 
As can be seen from the results, it takes only 2-3 seconds to generate proof of block possession. The process of generating such proofs is all at once. Once generated, these proofs can be made public along with the statement, and anyone can verify the correctness of the declared relations by evaluating the proofs. The time required to verify such succinct proofs is in the millisecond range and is, therefore, very efficient.

For the performance of Block System Possession, the proving time is 2.77 s, the verification time is 19 ms, and the proof size is 517 Byte. As can be seen from the results, it takes only 2-3 seconds to generate a proof of block possession. and about 10 seconds to generate a proof for block information exchange. 

\nop{
\begin{table}[]
\centering
\caption{\textbf{Performance of zero-knowledge proofs.}}
\resizebox{0.8\columnwidth}{!}{
\begin{tabular}{@{}ccc@{}}
\toprule
                  & Block Possession & Block Exchange \\ \midrule
Proving time      & 2.77 s           & 10.71 s        \\
Verification time & 19 ms            & 40 ms          \\
Proof size        & 517 Byte         & 612 Byte       \\ \bottomrule
\end{tabular}}

\label{tab:zkp}
\end{table}
}

\subsection{PSM-DoS Attack and The Promise of Secret Publication}
\label{sec:BlockRel}
It should be noted that the above mechanism can only prove to other miners that the attacker has indeed mined a block that satisfies the requirement. But there is no guarantee that the attacker will share the secret in the future. By potentially withholding the reserved block, the attacker could waste the attracted miner's mining power and may cause a DoS attack on the public branch. 
Since this new attack has to launch with PSM, we call it PSM-DoS attack.

The workflow of the PSM-DoS attack is shown in Figure~\ref{fig:PBDoS}. First, the attacker distributes the partial block data to all the miners and attracts attracted miners to join the attacker's private branch. In the meantime, the attacker leaves the private branch and puts all the mining power back into the public branch. Then, once attracted miners on the private branch find the new block, the attacker refuses to release the first block. In this case, the public chain will not accept the new block because its previous block is not published. Thus, miners who try to follow the attacker will suffer losses if the attacker consistently fails to disclose the full block to other miners. 

\begin{figure}[htbp]
	\centerline{\includegraphics[width=0.6\columnwidth]{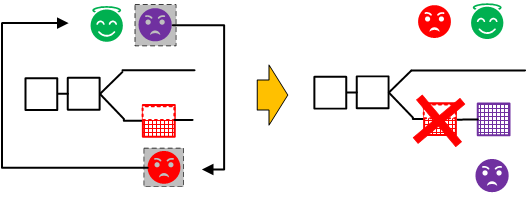}}
	\caption{\textbf{Workflow of PSM-DoS Attack.} After attracting the attracted miners to work on its private branch, the attacker discards the secret and goes back to work on the public branch.}
	\label{fig:PBDoS}
\end{figure}

The PSM attacker must implement countermeasures to mitigate the potential concerns of attracted miners regarding PSM-DoS attacks. Here, we present two examples: 



First, the attacker can address the concerns of attracted miners who discover the new block by carefully selecting the number of hidden bytes, ensuring that the calculation of these hidden bytes can be performed within an acceptable timeframe.
As we have demonstrated in Section \ref{sec:BlockPub}, when the attacker shares the partial block, miners receive the full block without the \textit{nonce} and several bytes of coinbase transaction information. Assuming that the attacker hides $n$ bytes of data, attracted miners need to calculate $2^{4n}$ times of hash to get the hidden data. 

Take Bitcoin as an example. According to~\cite{nakamoto2008bitcoin}, with difficulty $D$, we can approximate the hash power of all the miners in the network as $D\times \frac{2^{32}}{600}$ H/s. Assuming that attacker hides $b$ bytes of data so that $fr$ percent of miners need to calculate for a duration of $T_c$, then we have
\begin{small}
\begin{equation}
\begin{aligned}
T_c = \frac{2^{8b}}{fr\times D \times \frac{2^{32}}{600}},\     \ b = \log_{2^8}(T_c\times fr \times D\times \frac{2^{32}}{600}).
\end{aligned}
	\label{equ:bytes}
\end{equation}
\end{small}
 According to \cite{blockexplorer.com}, the Hashrate of the Bitcoin network is $1.55371 \times 10^{16}H/s$. The difficulty of bitcoin $D_{btc} \approx 31.25 \times 10^{12}$. By employing Equation (\ref{equ:bytes}), hiding 7 bytes of data enables a 1\% mining power miner to compute the hidden bytes in about 1 second. Similarly, hiding 8 bytes of data allows a 1\% mining power miner to determine the hidden bytes in about 1 minute. Though we take Bitcoin as the example, the proposed method works not only in Bitcoin but also in other PoW-based public blockchains (or Bitcoin-like) systems.


\nop{
\begin{figure}[htbp]
	\centering
	\subfigure[Number of bytes the attacker needs to hide when $T_C=10 min$.]{\label{fig:Diff}\includegraphics[width=0.49\columnwidth]{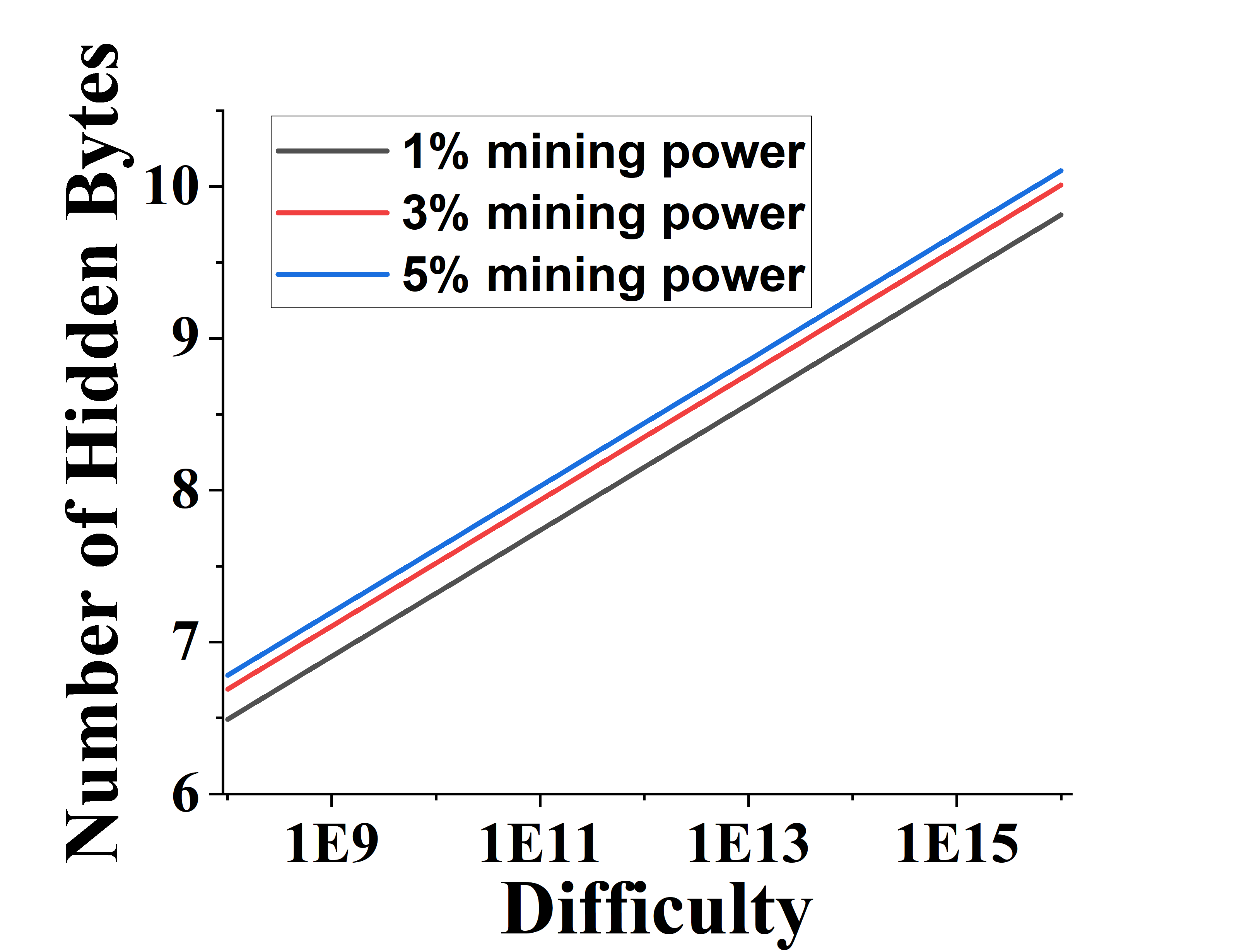}}
	\subfigure[The expected time to get the hidden data with $D = 31.25 \times 10^{12}$.]{\label{fig:Hidden}\includegraphics[width=0.49\columnwidth]{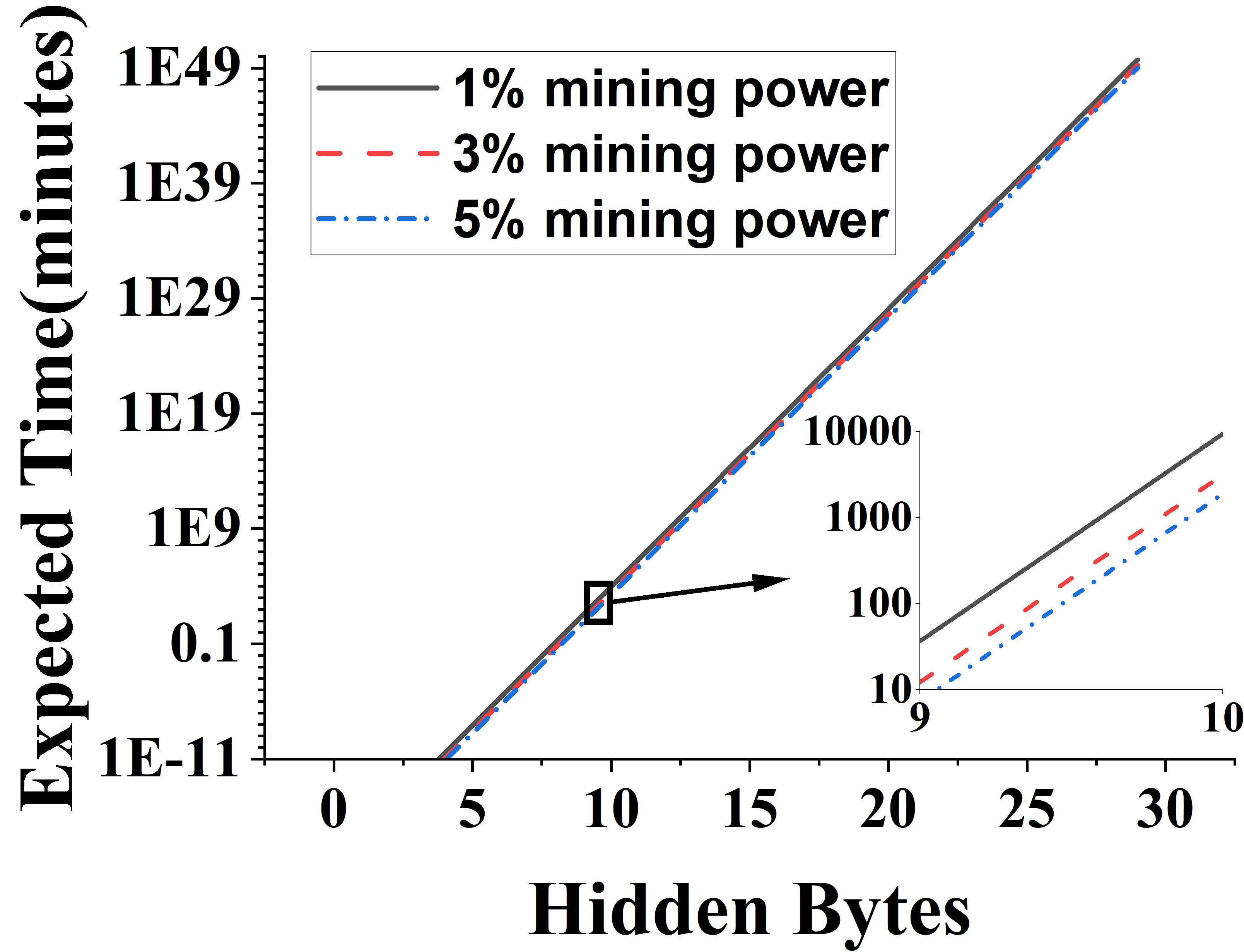}}
	\caption{\textbf{Evaluation results about the number of hidden bytes.}}
	\label{fig:HiddenByteEvaluation}
\end{figure}
}

Another method to incentivize miners and address concerns is to create a smart contract with collateral to cover rewards and costs incurred within smart contracts. The contract includes cryptographic commitment of block data. By doing so, the attacker effectively guarantees that it will redeem the collateral at some agreed point in the future by submitting a secret to the contract that satisfies the requirement and opens the commitment. Otherwise, this collateral will be transferred to the attracted miner who finds the new block. On the other hand, when the collateral value is less than the attacker's cost to mislead other miners, the miner has every incentive to provide an invalid block to deceive other miners. To solve the problem, the collateral value is required to be much larger than the cost of mining for the specified period. Specifically, the flow of this smart contract is as follows.

\begin{enumerate}
    \item The attacker pre-collateralizes the contract with some coins in the contract with a value equivalent to the proceeds of mining $n$ blocks.
    \item An attracted miner discovers a subsequent new block of the attacker's partial block and submits the information about the new block to the smart contract.
    \item The smart contract verifies the new block. If the verification passes, the contract opens a challenge period within which the attacker should disclose the partial block. If the attacker discloses the full information of the partial block within the challenge period, it can redeem all the collateral it has previously deposited. Suppose the attacker fails to disclose the block information within the challenge period or discloses an incorrect block. In that case, the contract transfers the attacker's collateral to the attracted miner who finds the new block.

\end{enumerate}


Assuming the challenge period lasts for a duration $T_C$, the miner will get the revenue, and the system will surely go back to a single branch state after $T_C$. If the attacker refuses to publish the full block, the attacker will lose both the partial block revenue and the collateral. The attacker can get more profits by launching PSM-DoS only if it could get more than $n+1$ blocks within $T_C$. 

 According to~\cite{karame2012double}, the possibility of finding a new block within duration $T$ can be expressed as:
\begin{small}
    \begin{equation}
     R(T)=\alpha_e \times (1-(1-p_e)^{\frac{T}{T_{avg}}}),
     \label{equ:RT}
    \end{equation}
\end{small}
    where the $T_{avg}$ is the average block generation time in the blockchain. In bitcoin $T_{avg}=10$ minutes.  $p_e = 64\% $, which means all miners in the network have a possibility of $64\%$ generating a new block within 10 minutes.

If the attacker can provide large enough collateral with a low enough $T_C$, then launching a PSM-DoS attack is economically not worthwhile.


The attacker has the option to employ either both methods or one of these approaches to address the concerns raised by the attracted miners.
Though our proposed method could address most of the concerns of rational miners, we agree that it is not realistic to assume all the miners are rational in practice. In the following, we discuss the attacker's gain with different ratios of the rational miners.

    
\section{PSM Analysis and Comparisons}\label{sec:simulation}
When calculating the revenue, an attacker can estimate $\alpha_A$, $\gamma$, and mining power distribution in a blockchain. It can also roughly estimate the number of rational miners by the number of applications (e.g., via smart contract) for partial block data. However, it cannot know the fraction of rational miners in the entire network. Thus, the attacker needs to estimate its reward with a different fraction of rational miners. Our assumptions here are consistent with those in recent studies~\cite{ozisik2017estimation,saad2021syncattack}.

In this section, we use numeric analysis to evaluate the reward of the PSM strategy. 


\subsection{Comparison with Honest Mining}
\label{sec:HonestMin}

\subsubsection{Quantitative Analysis}
We mathematically analyze the revenue of PSM against honest mining. As stated before, the attacker's computational power is $\alpha_A$. If the attacker chooses to follow the honest mining strategy, the possibility of getting the second block revenue is $\alpha_A$. Overall, the attacker's expected profit when following honest mining is:
\begin{small}
\begin{equation}
    R^A_H =\alpha_A.
    \label{equ:Honest}
\end{equation}
\end{small}
If the attacker chooses the PSM attack, it could attract rational miners with $\alpha_i$ mining power, and the total mining power of public miners is $\alpha_h=1-\alpha_A-\alpha_i$. For the attacker, the expected RER of following PSM instead of honest mining can be expressed as follows:

\begin{small}
\begin{equation}
\begin{aligned}
    &RER_{A}^{P,H} =\frac{R^A_P-R^A_H}{R^A_H}\\ =&\frac{(\aA\ah+\aA+1)\ai}{(2\aA^2\ah+2\aA^2+\aA)\ai+2\aA^2\ah^2+(2\aA^3+\aA)\ah+2\aA^3}-1.
    \label{equ:RevHonest}
\end{aligned}
\end{equation}
\end{small}

\begin{figure*}[ht]
	\centering
	\subfigure[Rushing ability $\gamma$=0]{\label{SHr0}\includegraphics[width=0.65\columnwidth]{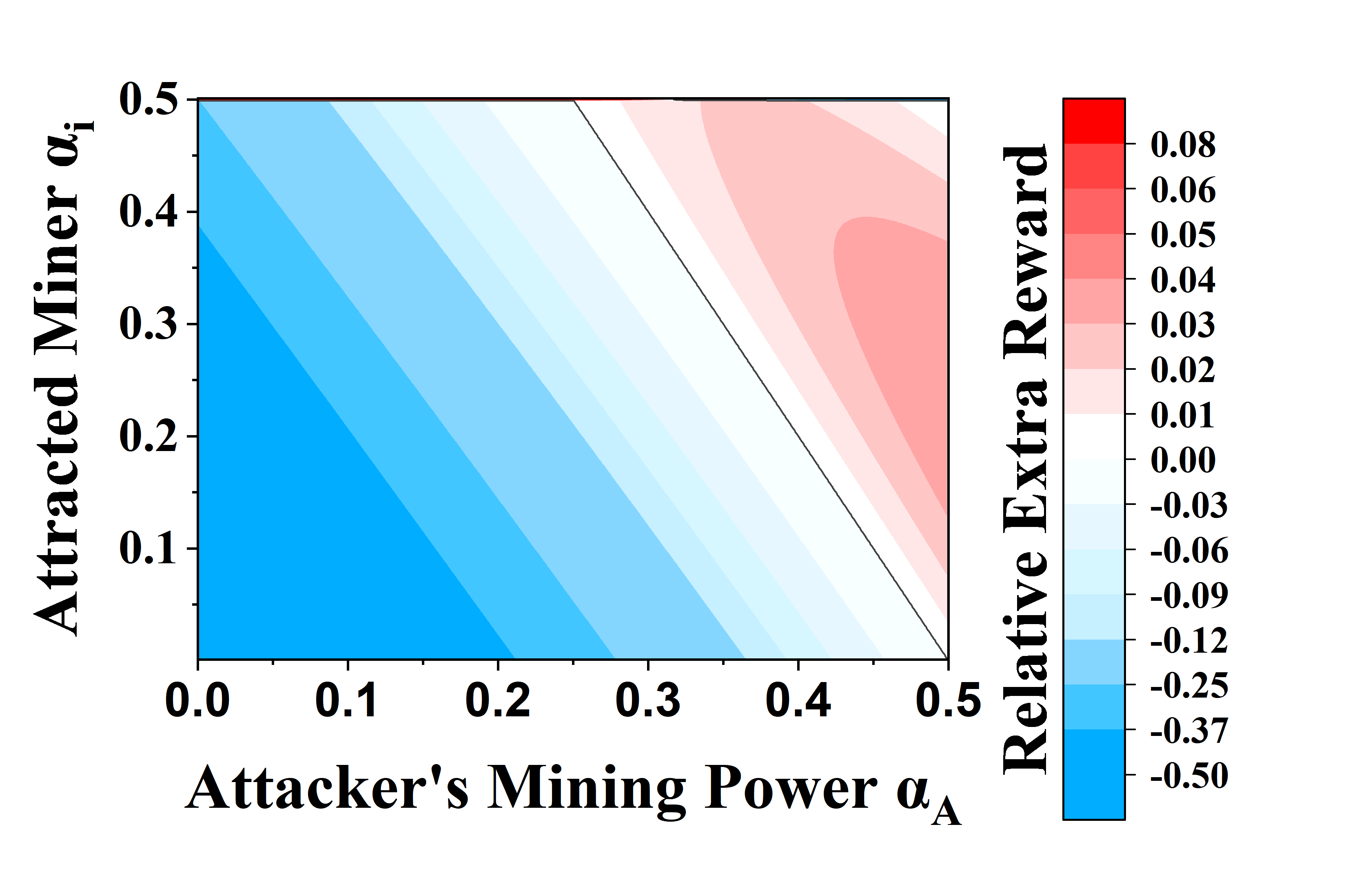}}
	\quad
	\subfigure[Rushing ability $\gamma$=0.5]{\label{SHr05}\includegraphics[width=0.65\columnwidth]{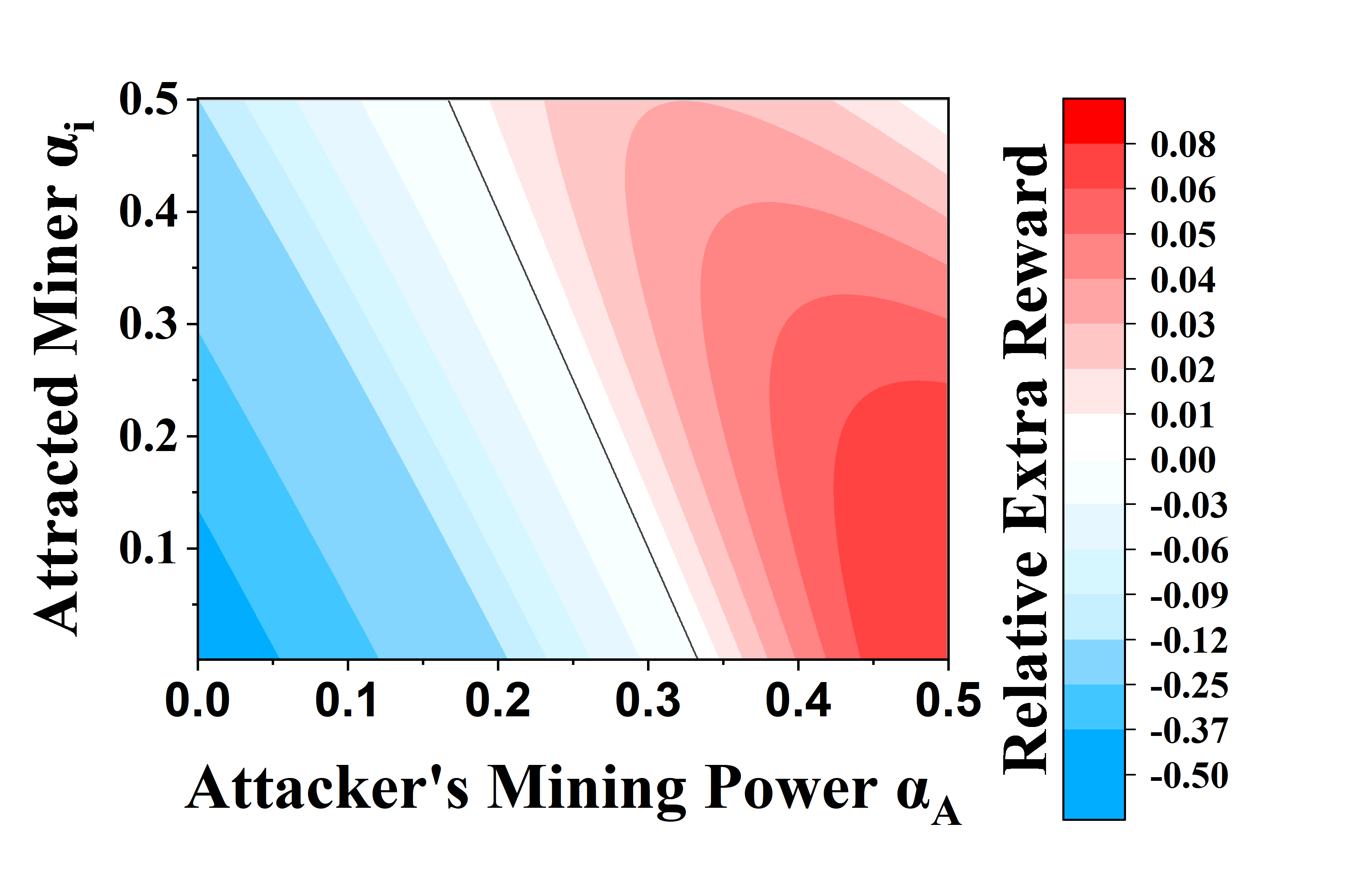}}
	\quad
	\subfigure[Rushing ability $\gamma$=1]{\label{SHr1}\includegraphics[width=0.65\columnwidth]{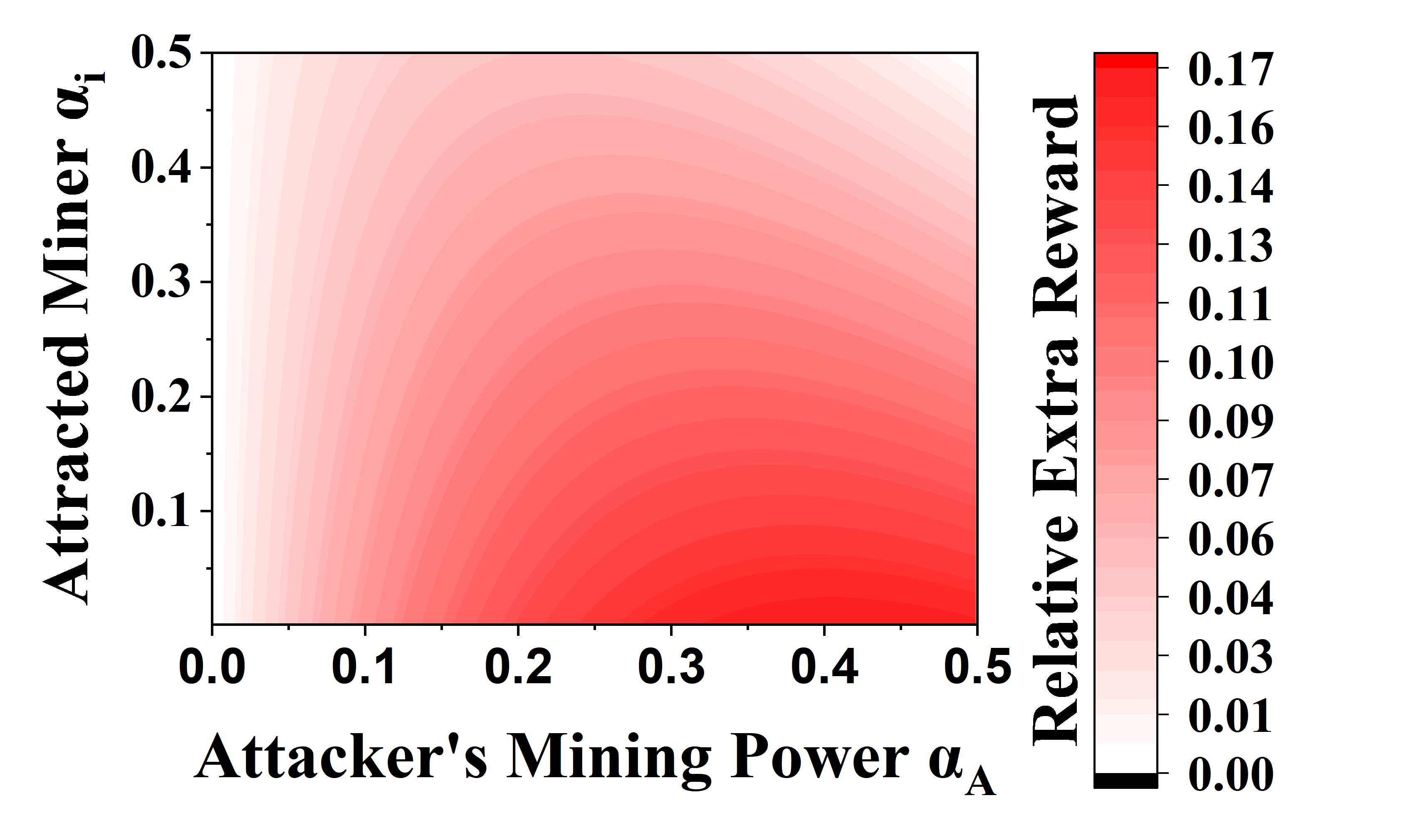}}
	\caption{\textbf{Attacker's relative extra reward when choosing PSM instead of honest mining ($RER_{A}^{P,H}$).} The solid line represents no extra reward.}
	\label{fig:RERSelfHone}
\end{figure*}

In Figure \ref{fig:RERSelfHone}, we show the numerous simulation results of RER of the PSM strategy over the honest mining strategy. The higher the $\gamma$ is, the more likely the attackers can get more rewards than honest mining. When $\gamma=1$, PSM can surely get more rewards than honest mining. If it can attract enough rational miners, the PSM attacker with enough mining power can still get more rewards than honest mining when $\gamma=0$. 

\subsubsection{Simulation Results}
To further verify the accuracy of our quantitative results, we implement a Monte Carlo simulator in Java to verify our theoretical analysis. We simulate an attacker with $\alpha_A=0.2$ and run the simulator over $10^9$ rounds. The upper bound for error is $10^{-4}$. The results are shown in Table \ref{tab:HonestRER}. The attacker's RER is the same as expected.

\begin{table}[]
\caption{\textbf{Monte Carlo simulation results of attacker's RER when choosing PSM instead of honest mining.}}
\label{tab:HonestRER}
\centering
\resizebox{\columnwidth}{!}{
\begin{tabular}{@{}clllll@{}}
\toprule
\multicolumn{1}{l}{\diagbox{$\gamma$}{$\alpha_i$}}  & \multicolumn{1}{c}{0.1} & \multicolumn{1}{c}{0.2} & \multicolumn{1}{c}{0.3} & \multicolumn{1}{c}{0.4} &   \\ \midrule
0                             & -29.167(-29.17)         & -20.0(-20.0)            & -12.5(-12.5)            & -6.67(-6.67)                       \\
0.25                          & -18.96(-18.96)          & -12.5(-12.5)            & -7.29(-7.29)            & -3.33(-3.33)                      \\
0.5                           & -8.75(-8.75)            & -5.0(-5.0)              & -2.08(-2.08)            & 0.0(0.0)                             \\
0.75                              & 1.46(1.46)              & 2.5(2.5)                & 3.13(3.12)              & 3.33(3.33)                          \\
1                          & 11.67(11.67)            & 10.0(10.0)              & 8.33(8.33)              & 6.67(6.67)                          \\ \bottomrule
\end{tabular}}
\end{table}

\subsection{Comparison with Selfish Mining\label{sec:SelfMin}}


\subsubsection{Quantitative Analysis}

According to~\cite{9eyal2014majority}, if choosing selfish mining, the attacker's reward is:

\begin{small}
\begin{equation}
R^A_S=\frac{\aA(1-\aA)^2(4\aA+\ar(1-2\aA))-\aA^3}{1-\aA(1+(2-\aA)\aA)},
\label{equ:RevSelfish}
\end{equation}
\end{small}

and the expected RER of PSM over selfish mining is

\begin{small}
\begin{equation}
\begin{aligned}
    &RER_{A}^{P,S}=\frac{\aA^3-2\aA^2-\aA+1}{\aA+1}\times \\
    &(\frac{(\ar-1)\ai^2+(2(\aA-1)-(3\aA-2)\ar)\ai}{(2\aA^3-5\aA^2+4\aA-1)\ar-4\aA^3+9\aA^2-4\aA}+\\
    &\frac{(\aA-1)^2\ar+4\aA-2\aA^2}{(2\aA^3-5\aA^2+4\aA-1)\ar-4\aA^3+9\aA^2-4\aA})-1
    \label{eq:RAS}
    \end{aligned}
\end{equation}
\end{small}

\begin{figure*}[ht]
	\centering
	\subfigure[Rushing ability $\gamma$=0]{\label{SOr0}\includegraphics[width=0.65\columnwidth]{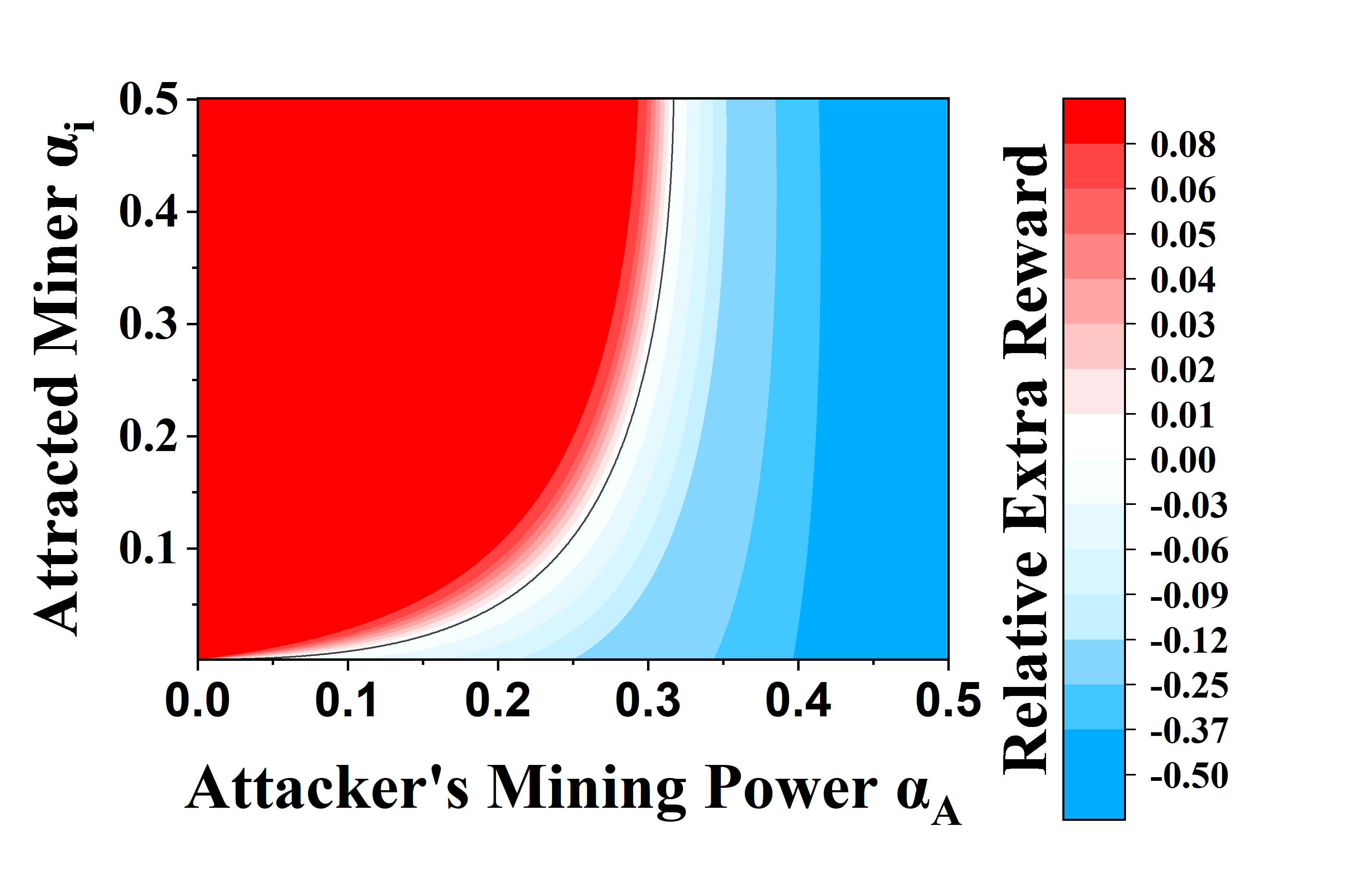}}
	\quad
	\subfigure[Rushing ability $\gamma$=0.5]{\label{SOr05}\includegraphics[width=0.65\columnwidth]{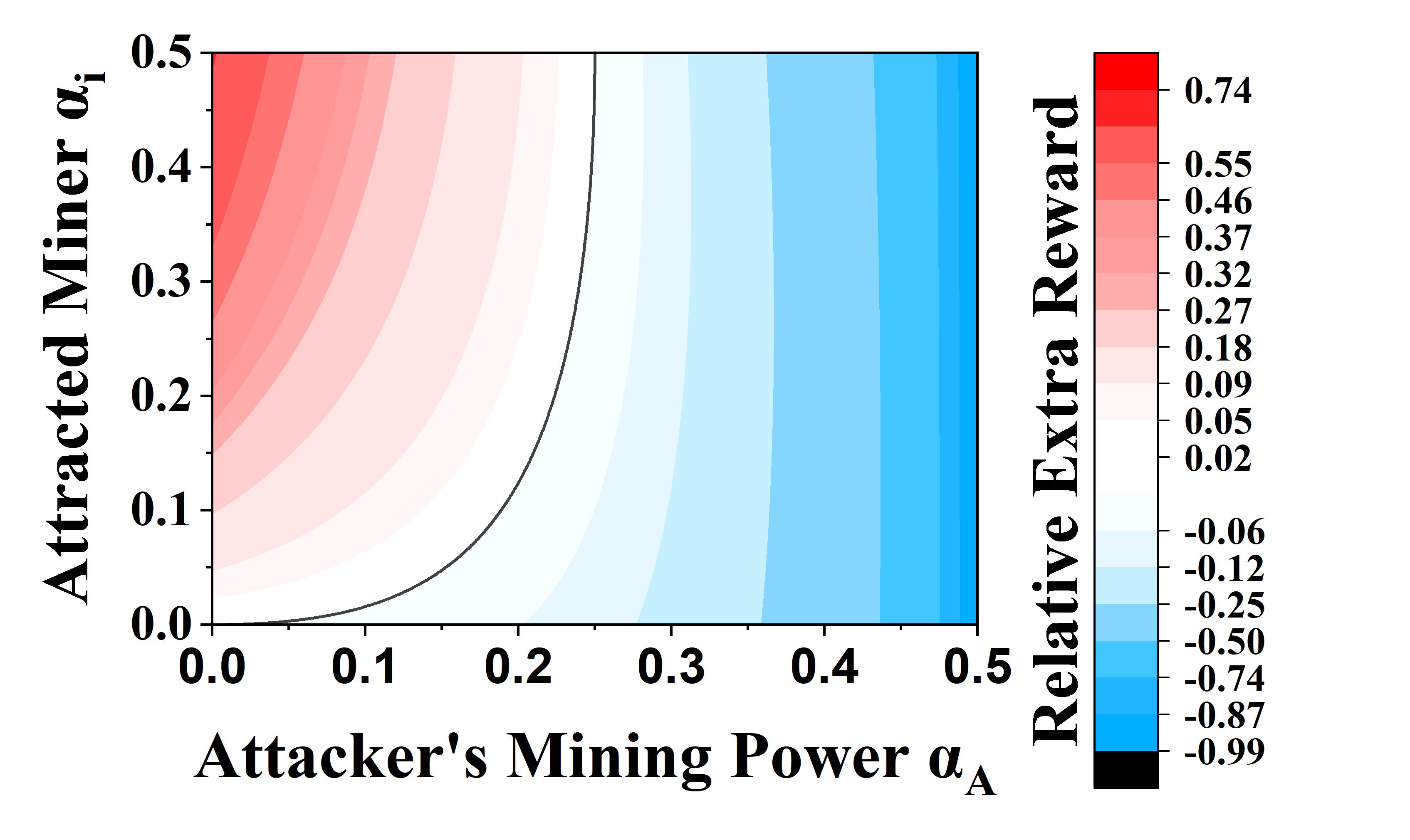}}
	\quad
	\subfigure[Rushing ability $\gamma$=1]{\label{SOr1}\includegraphics[width=0.65\columnwidth]{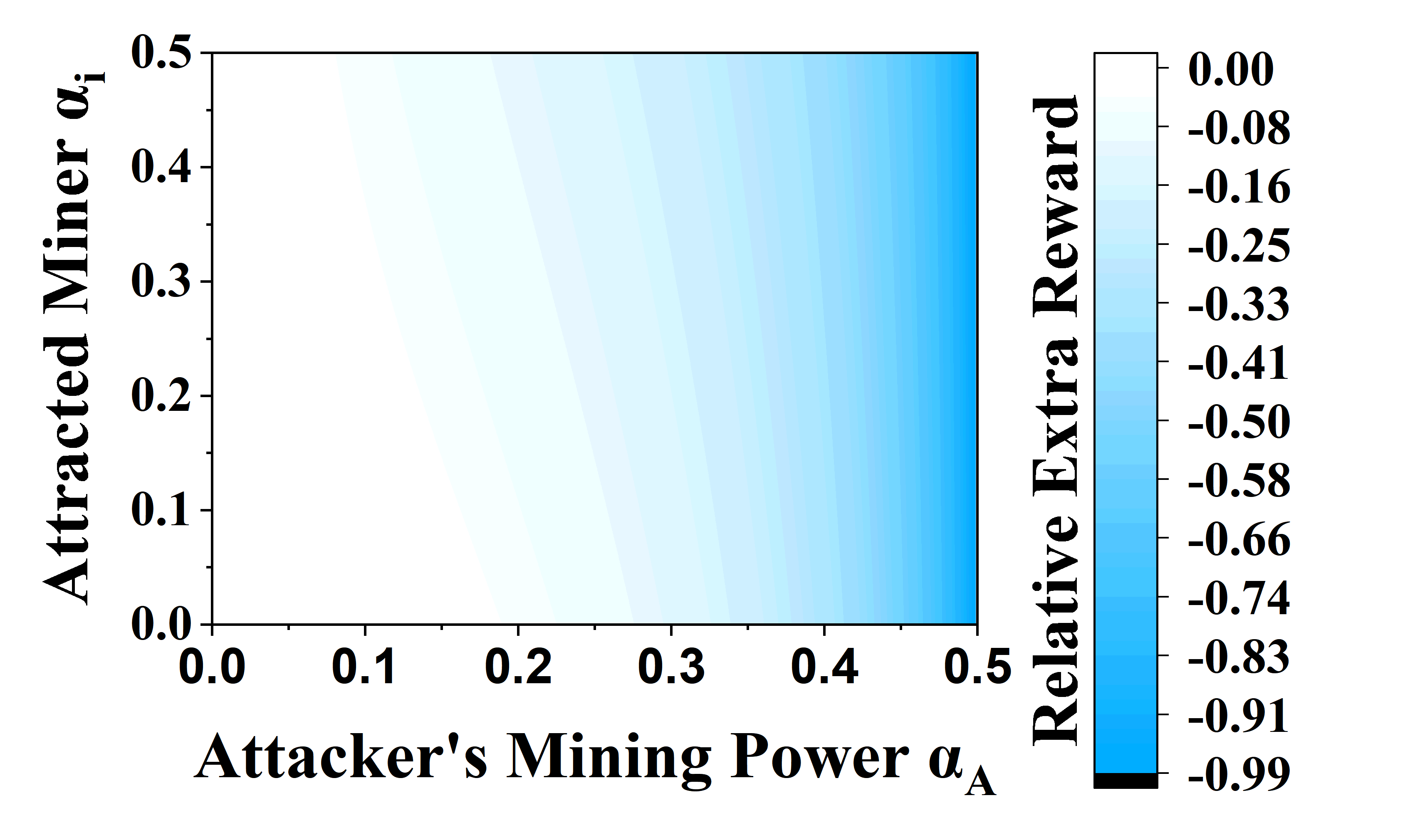}}
	\caption{\textbf{Attacker's relative extra reward when choosing PSM instead of selfish mining ($RER_{A}^{P,S}$).} The solid line represents no extra reward.}
	\label{fig:RERSelfOrig}
\end{figure*}

In Figure~\ref{fig:RERSelfOrig}, we show the numerous simulation results of the attacker's RER following PSM instead of the selfish mining strategy. The PSM attacker can get a higher reward than the selfish miner when its mining power is relatively small. When the attacker's mining power is large enough, the possibility of finding more than one block on its private branch becomes non-negligible. Thus, the revenue of selfish mining is higher than PSM. We propose an advanced PSM strategy to address this issue in Section~\ref{sec:APSM}.


\subsubsection{Simulation Results}
To further verify the accuracy of our quantitative results, assuming the attacker with a computation power of 0.2, we compare the Monte Carlo simulation results of the RER of PSM over selfish mining with our evaluation results. We run the Java-based simulator over $10^9$ rounds. The upper bound for error is $10^{-4}$. The Monte Carlo simulation results are shown in Table~\ref{tab:SelfRER}. The attacker's RER is the same as expected.

\begin{table}[]
\caption{\textbf{Monte Carlo simulation results of attacker's RER when choosing PSM instead of selfish mining.}}
\label{tab:SelfRER}
\centering
\resizebox{\columnwidth}{!}{
\begin{tabular}{@{}cllll@{}}
\toprule
\multicolumn{1}{l}{\diagbox{$\gamma$}{$\alpha_i$}}   & \multicolumn{1}{c}{0.1} & \multicolumn{1}{c}{0.2} & \multicolumn{1}{c}{0.3} & \multicolumn{1}{c}{0.4} \\ \midrule
0   & 8.07(8.05)   & 22.04(22.03) & 33.51(33.47) & 42.35(42.37)     \\
0.25  & 2.74(2.73)   & 10.93(10.92) & 17.52(17.52) & 22.54(22.54)    \\
0.5  & -1.06(-1.05) & 3.01(3.01)   & 6.2(6.17)    & 8.42(8.43)       \\
0.75  & -3.88(-3.88) & -2.9(-2.89)  & -2.3(-2.3)   & -2.11(-2.11)     \\
1   & -6.08(-6.07) & -7.48(-7.48) & -8.88(-8.88) & -10.29(-10.28) \\\bottomrule
\end{tabular}}
\end{table}

\section{Advanced PSM Strategy}
\label{sec:APSM}
Under the condition that the attacker can have the up-to-date block height from the public chain, we propose an optimized PSM strategy named Advanced PSM (A-PSM), which can further increase the profits of PSM for the attacker.

\subsection{Attack Overview}


When miners find new blocks, the A-PSM attacker follows the selfish-mining-like strategy to publish the partial-released block instead of simply releasing the partial blocks' secrets. The attacker will keep the secret private until the lead of the private branch is no more than 2 blocks. Attracted miners can immediately release the block it finds, but it will not be recognized as a valid block until the secret of the prior partial block is released.

Specifically, if a public miner finds a new block, two possible cases may happen: if the lead of the private branch is 2, the attacker will release all the partial blocks, and the system goes back to the single branch state. If the private branch's lead is more than 2, then the attacker and rational miners will continue working on its private branch until the lead is 2. 

To avoid cases where the private branch's length grows faster than the public branch, we extend our assumption that in the A-PSM scenario, the mining power of the attacker, together with rational miners, is no more than 50\%. 

We also extend the assumption that the attacker can promise that it will release the secret based on the length of both the private and the public branches. This assumption is reasonable because when launching mining attacks, attackers are motivated to maximize their revenue, and the secret computation mechanism in Section~\ref{sec:PracticalConcerns} further assures the malicious behavior is economically not worthwhile for the attacker. To further address the rational miner's concern, the attacker can also have the latest block height of the public branch via decentralized oracle~\cite{zhang2020deco} or reported by an attracted miner.

\begin{figure}[htbp]
	\centerline{\includegraphics[width=\columnwidth]{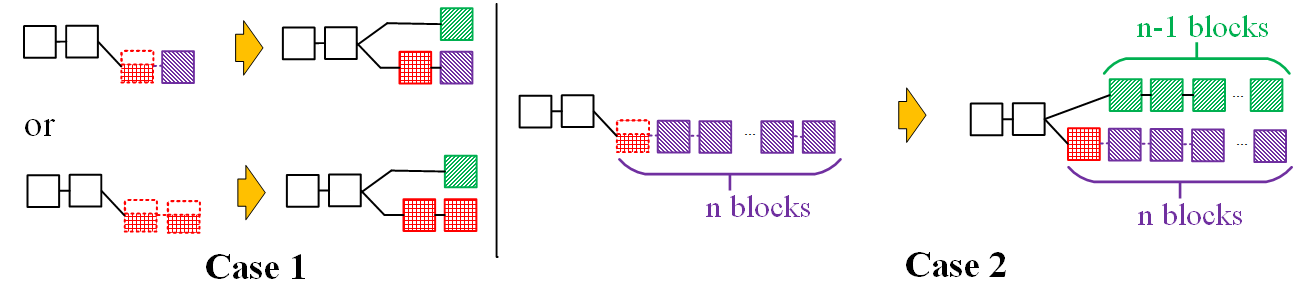}}
	\caption{\textbf{Partial block sharing with A-PSM strategy.}} 
	\label{fig:RealSelf}
\end{figure}

\subsection{A-PSM Reward}

Overall, if choosing the A-PSM strategy, the attacker and attracted miners' expected profits can be derived as follows:

\begin{theorem}
The profit of an A-PSM attacker is:
    \begin{small}
\begin{equation}
\begin{aligned}
&R^A_{AP}=\frac{2\aA\ai^3+(8\aA^2-5\aA)\ai^2+(10\aA^3-14\aA^2+2\aA)\ai}{\aA\ai^2+(2\aA^2-2\aA-2)\ai+\aA^3-2\aA^2-\aA+1}+\\
&\frac{(-\aA(\ai+\aA-1)^2(2\ai+2\aA-1))\ar+4\aA^4-9\aA^3+4\aA^2}{\aA\ai^2+(2\aA^2-2\aA-2)\ai+\aA^3-2\aA^2-\aA+1}.
\label{equ:APSMRev}
\end{aligned}
\end{equation}
\end{small}
\label{THM:APSMAtk}
\end{theorem}

\begin{theorem}
The attracted miner's expected profit is

\begin{small}
\begin{equation}
R^i_{AP}=\frac{2\aA\ai^3+(4\aA^2-4\aA-2)\ai^2+(2\aA^3-4\aA^2+1)\ai}{\aA\ai^2+(2\aA^2-2\aA-2)\ai+\aA^3-2\aA^2-\aA+1}.
\end{equation}
\end{small}

\label{THM:APSMGmn}
\end{theorem}

The proof of Theorem~\ref{THM:APSMAtk} and Theorem~\ref{THM:APSMGmn} is in Appendix ~\ref{sec:APSMattacker}.

\subsection{Profit Analysis}

\textbf{Rational Miners' profits Analysis:} For rational miner $k$, when following the public mining strategy, we can derive the following theorems:
\begin{theorem}
    When following the public mining strategy, the rational miner $k$'s expected profit' is: 
\begin{small}
    \begin{equation}
\begin{aligned}
   R^{k,APSM}_P=&\frac{(\aA(1-2\aA)\ak(\ak+\aA-1))\ar}{\aA^3-2\aA^2-\aA+1}\\
   &+\frac{(4\aA^3-6\aA^2+1)\ak}{\aA^3-2\aA^2-\aA+1}
    \end{aligned}
\end{equation}
\end{small}

The RER of following the A-PSM instead of public mining is:
\begin{small}
\begin{equation}
RER^{AP,H}_{k}=\frac{((-\aA\ai)-\aA^2+\aA)\ar-\aA\ai+\aA^2}{(\aA\ai+\aA^2-\aA)\ar-2\aA^2+2\aA+1}.
\label{equ:APSMMinerRER}
\end{equation}
\end{small}
\label{THM:APSMMiner}
\end{theorem}
The proof of Theorem~\ref{THM:APSMMiner} is in Appendix~\ref{sec:APSMMiner}.

In Figure~\ref{fig:RERCompHeat}, we show numerous simulation results of miners' RER when following the greedy mining strategy instead of the public mining strategy with an A-PSM attacker. 


\textbf{Attacker's Profit Analysis:} When following the honest mining strategy, the attacker's revenue is shown in Equation (\ref{equ:Honest}). The RER of the attacker following A-PSM rather than an honest mining strategy is: 

\begin{small}
\begin{equation}
\begin{aligned}
&RER^{AP,H}_{A}= \\
&\frac{(-\aA(\ai+\aA-1)^2(2\ai+2\aA-1))\ar}{\aA^2\ai^2+(2\aA^3-2\aA^2-2\aA)\ai+\aA^4-2\aA^3-\aA^2+\aA}+\\
&\frac{2\aA\ai^3+(8\aA^2-5\aA)\ai^2+(10\aA^3-14\aA^2+2\aA)\ai}{\aA^2\ai^2+(2\aA^3-2\aA^2-2\aA)\ai+\aA^4-2\aA^3-\aA^2+\aA}+\\
&\frac{4\aA^4-9\aA^3+4\aA^2}{\aA^2\ai^2+(2\aA^3-2\aA^2-2\aA)\ai+\aA^4-2\aA^3-\aA^2+\aA}-1.
\label{equ:APSMHonAtk}
\end{aligned}
\end{equation}
\end{small}

The reward of following the selfish mining strategy is shown in Equation (\ref{equ:RevSelfish}). The RER of A-PSM over selfish mining strategy is: 
\begin{small}
\begin{equation}
\begin{aligned}
&RER^{AP,S}_{A}=\frac{1-\aA(1+(2-\aA)\aA)}{\aA(1-\aA)^2(4\aA+\ar(1-2\aA))-\aA^3}\times\\
&(\frac{(-\aA(\ai+\aA-1)^2(2\ai+2\aA-1))\ar+4\aA^4-9\aA^3+4\aA^2}{\aA\ai^2+(2\aA^2-2\aA-2)\ai+\aA^3-2\aA^2-\aA+1}+\\
&\frac{2\aA\ai^3+(8\aA^2-5\aA)\ai^2+(10\aA^3-14\aA^2+2\aA)\ai}{\aA\ai^2+(2\aA^2-2\aA-2)\ai+\aA^3-2\aA^2-\aA+1})
\label{equ:APSMSelfAtk}
\end{aligned}
\end{equation}
\end{small}

\begin{figure}[ht]
	\centering
	\subfigure[Rushing ability $\gamma$=0]{\label{compminerr0}\includegraphics[width=0.45\columnwidth]{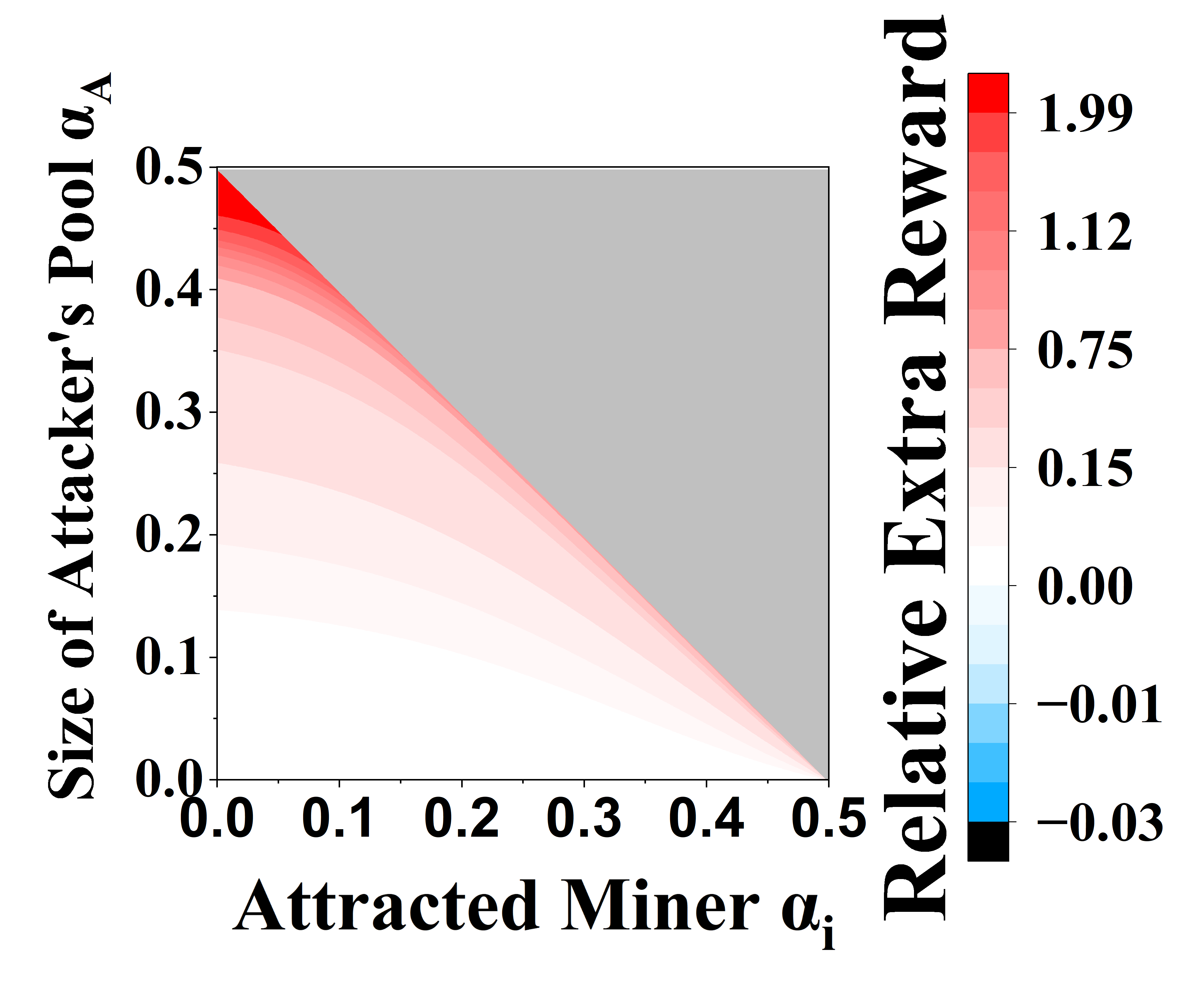}}
	\subfigure[Rushing ability $\gamma$=1]{\label{compminerr1}\includegraphics[width=0.45\columnwidth]{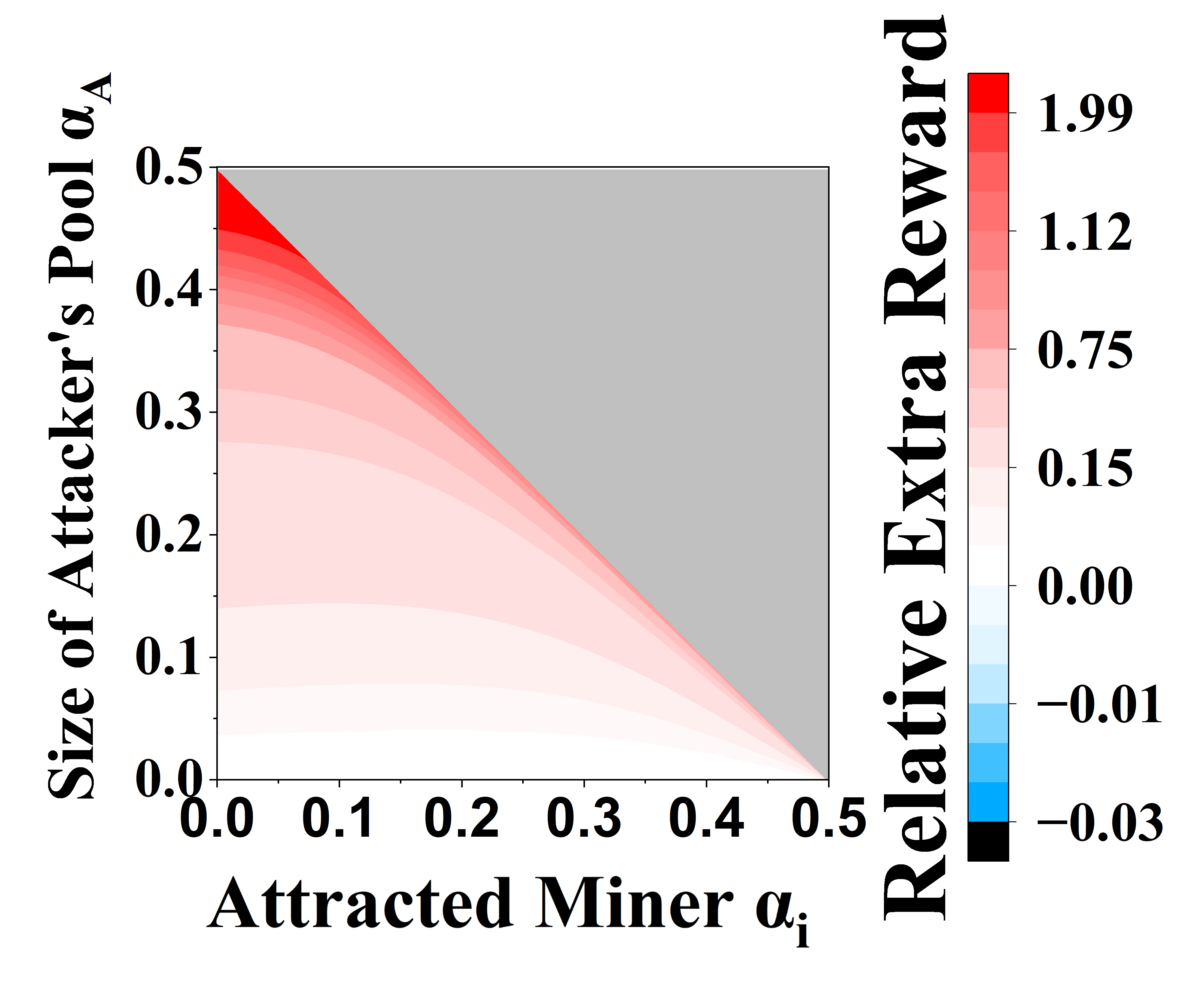}}
	\caption{\textbf{The relative extra rewards of the rational miner when choosing the greedy mining strategy instead of public mining strategy} with an A-PSM attacker ($RER^{AP,H}_{i}$). The solid line represents no extra reward.}
	\label{fig:RERCompHeat}
\end{figure}

\begin{figure}[ht]
	\centering
	\subfigure[$RER^{AP,H}_{A},\gamma=0$]{\label{compHr0}\includegraphics[width=0.45\columnwidth]{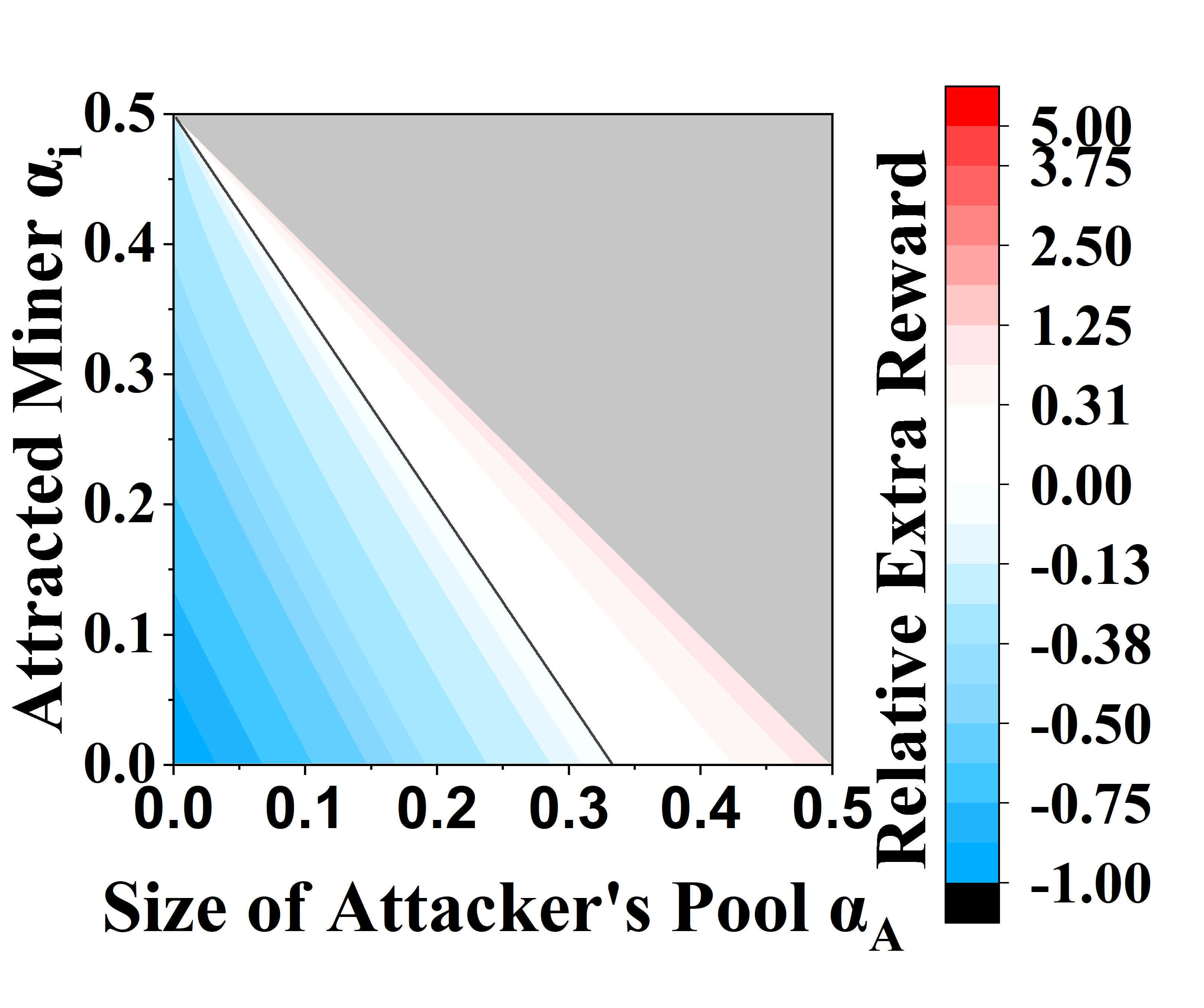}}
	\subfigure[$RER^{AP,H}_{A},\gamma=1$]{\label{compHr1}\includegraphics[width=0.45\columnwidth]{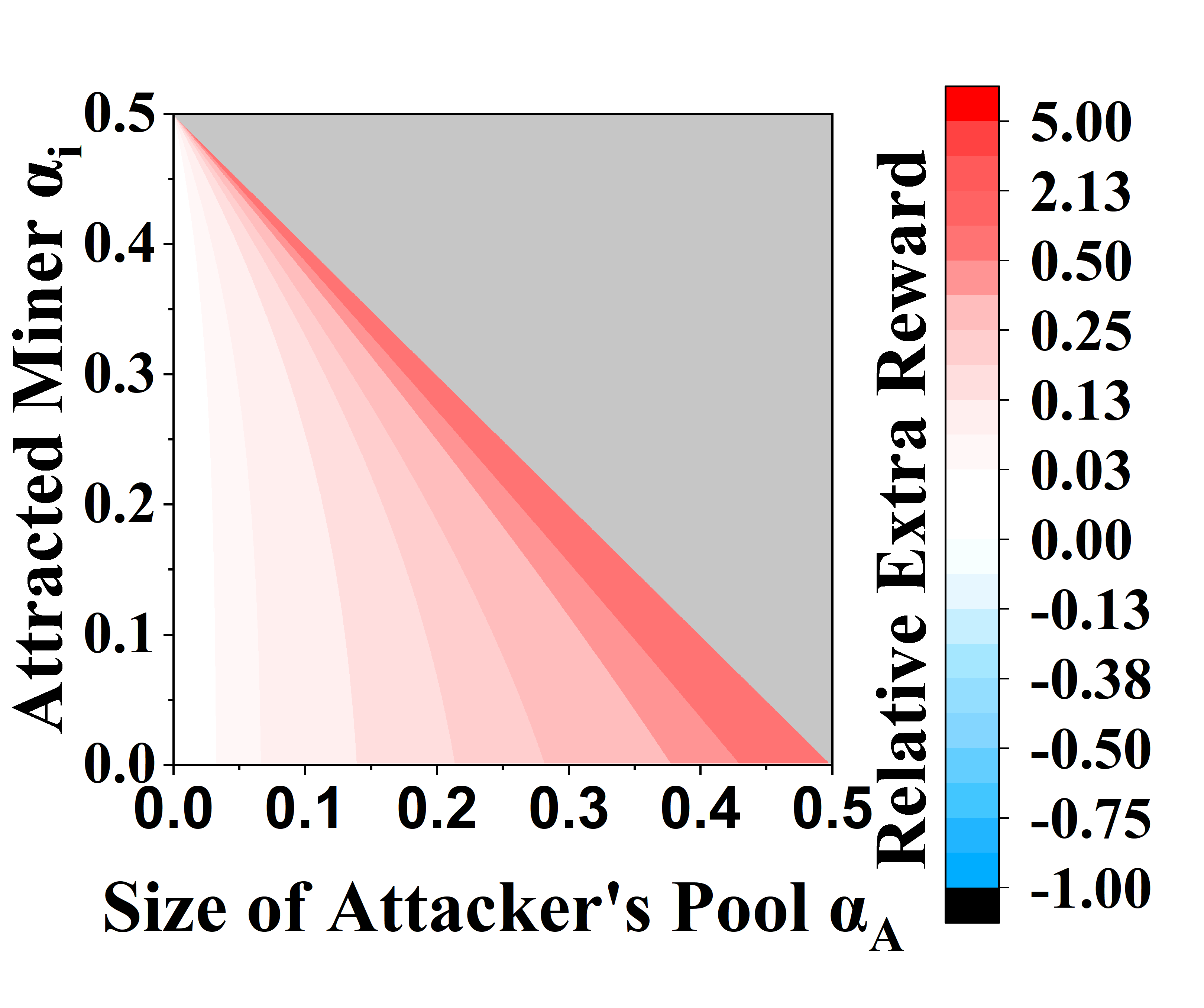}}
    	\subfigure[$RER^{AP,S}_{A},\gamma=0$]{\label{compSr0}\includegraphics[width=0.45\columnwidth]{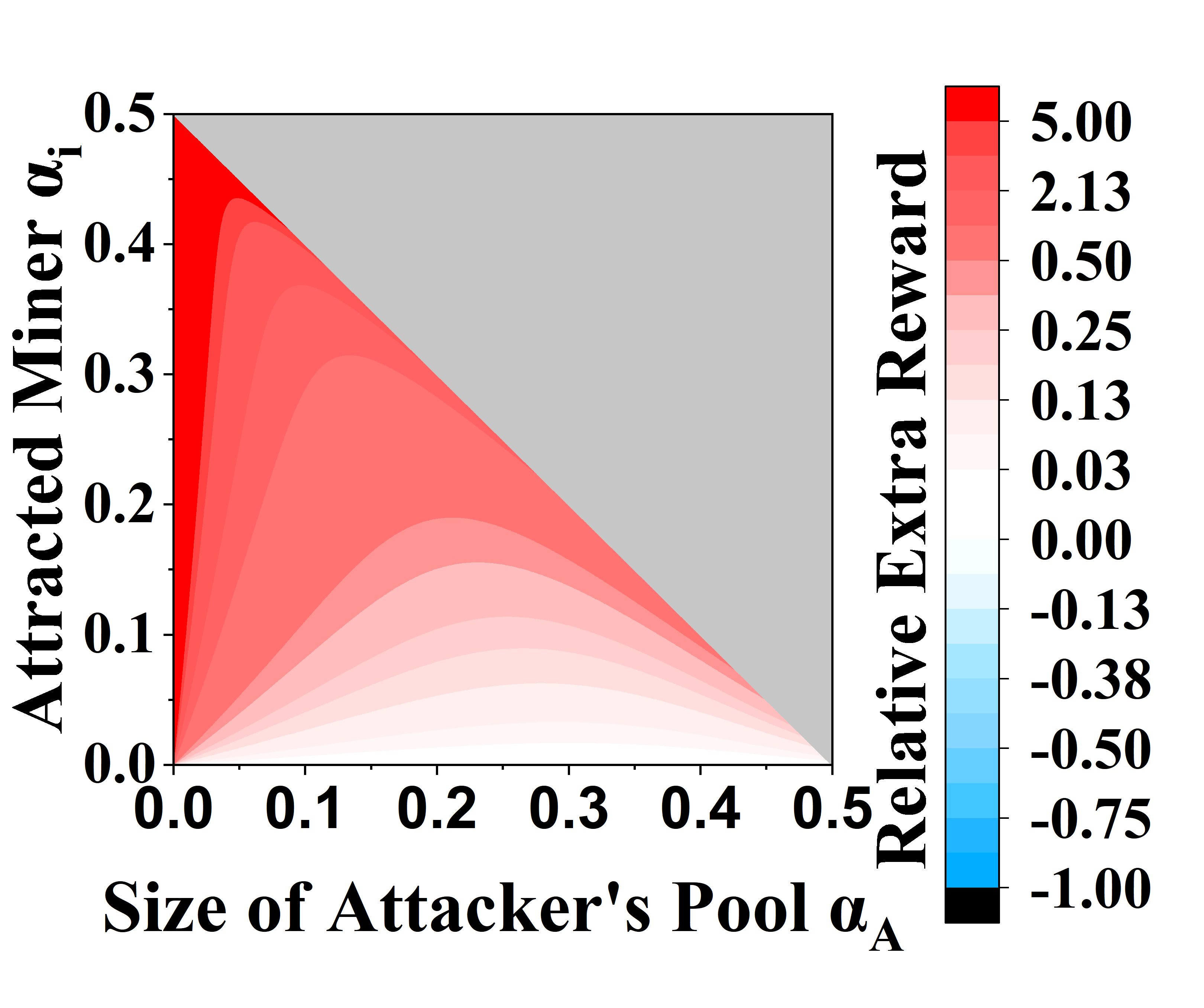}}
	\subfigure[$RER^{AP,S}_{A},\gamma=1$]{\label{compSr1}\includegraphics[width=0.45\columnwidth]{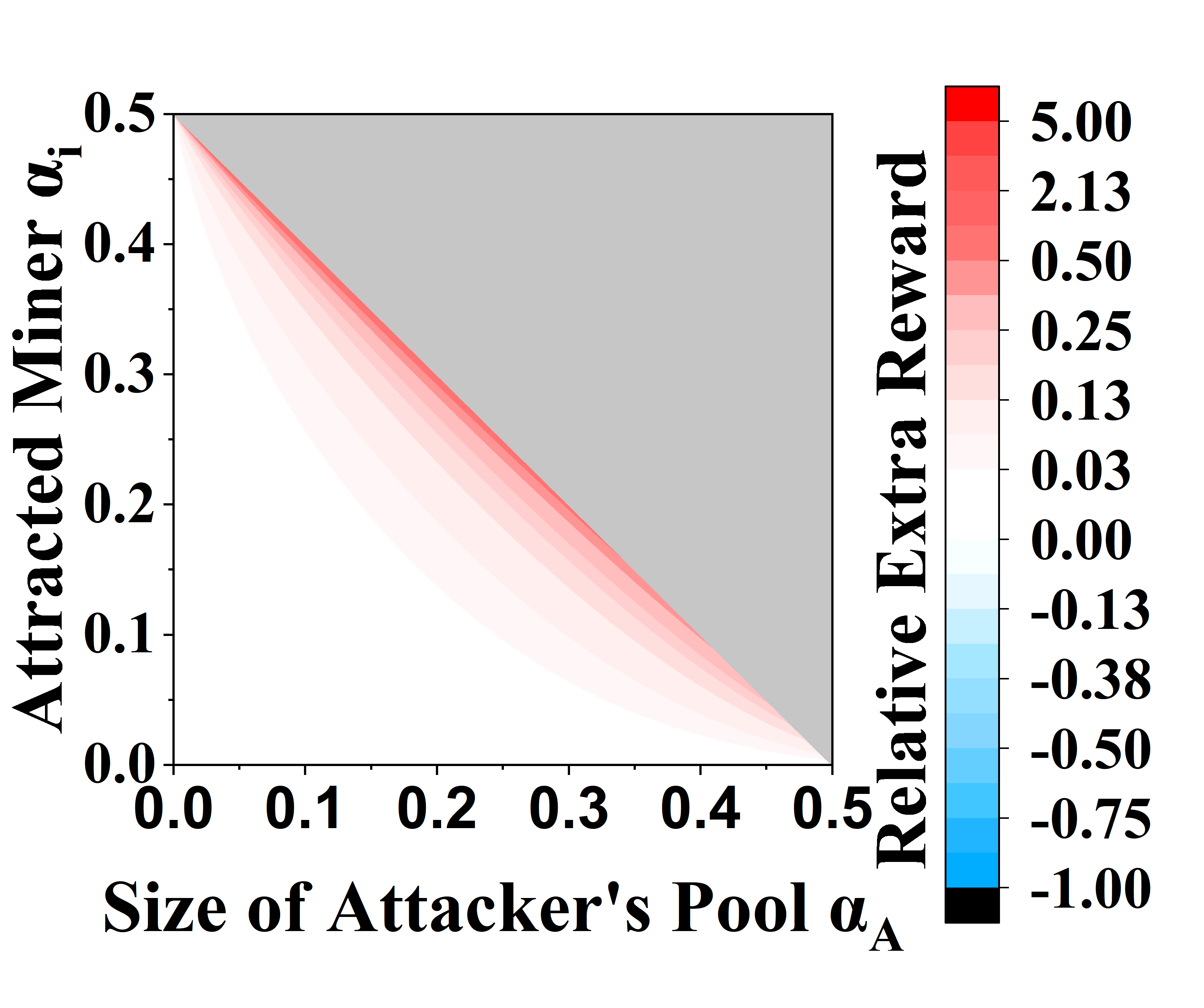}}
	\caption{\textbf{Attacker's relative extra reward of A-PSM over honest mining and selfish mining.} The solid line represents no extra reward.}
	\label{fig:CompRERMinerQ}
\end{figure}

In Figure~\ref{compHr0} and~\ref{compHr1}, we show the simulation results of the attacker's RER when following A-PSM instead of the honest mining strategy with $\gamma =0$ and $1$, respectively. The results of the attacker's RER of A-PSM over selfish mining are shown in Figure~\ref{compSr0} and~\ref{compSr1}  with $\gamma =0$ and $1$, respectively.

We compare the Monte Carlo simulation results of the A-PSM attack with honest mining and selfish mining simultaneously. In the simulation, we show the profits for the attacker with a computation power of 0.1 over $10^9$ rounds. The upper bound for error is $10^{-4}$. The comparison results with honest mining are shown in Table~\ref{tab:CompHonestRER}, and the comparison results with selfish mining are shown in Table~\ref{tab:CompSelfRER}.

\begin{table}[]
\caption{\textbf{Monte Carlo simulation results of attacker's RER when choosing A-PSM instead of honest mining.}}
\label{tab:CompHonestRER}
\centering
\resizebox{\columnwidth}{!}{
\begin{tabular}{@{}cllll@{}}
\toprule
\multicolumn{1}{l}{\diagbox{$\gamma$}{$\alpha_i$}}   & \multicolumn{1}{c}{0} & \multicolumn{1}{c}{0.1} & \multicolumn{1}{c}{0.2} & \multicolumn{1}{c}{0.3} \\ \midrule
0    & -64.31(-64.32)        & -47.88(-47.88)          & -31.37(-31.36)          & -9.09(-9.09)            \\
0.25 & -45.91(-45.91)        & -33.33(-33.33)          & -20.23(-20.23)          & -0.91(-0.91)            \\
0.5  & -27.5(-27.5)          & -18.78(-18.79)          & -9.09(-9.09)            & 7.28(7.27)              \\
0.75 & -9.1(-9.09)           & -4.23(-4.24)            & 2.05(2.05)              & 15.45(15.45)            \\
1    & 9.32(9.32)            & 10.3(10.3)              & 13.18(13.18)            & 23.64(23.64)            \\ \bottomrule
\end{tabular}}
\end{table}

\begin{table}[]

\caption{\textbf{Monte Carlo simulation results of attacker's relative extra reward of A-PSM over selfish mining.}}
\label{tab:CompSelfRER}
\centering
\resizebox{\columnwidth}{!}{
\begin{tabular}{@{}cllll@{}}
\toprule
\multicolumn{1}{l}{\diagbox{$\gamma$}{$\alpha_i$}}   & \multicolumn{1}{c}{0} & \multicolumn{1}{c}{0.1} & \multicolumn{1}{c}{0.2} & \multicolumn{1}{c}{0.3} \\ \midrule
0    & -0.01(0.0)            & 46.04(46.07)            & 92.39(92.36)            & 154.8(154.78)           \\
0.25 & 0.02(0.0)             & 23.25(23.25)            & 47.47(47.48)            & 83.17(83.19)            \\
0.5  & -0.04(0.0)            & 12.01(12.02)            & 25.39(25.39)            & 48.02(47.96)            \\
0.75 & 0.0(0.0)              & 5.32(5.33)              & 12.26(12.25)            & 27.0(27.0)              \\
1    & 0.01(0.0)             & 0.9(0.9)                & 3.54(3.53)              & 13.09(13.1)             \\ \bottomrule
\end{tabular}}
\end{table}

\section{Optimal Mining Strategy}
\label{sec:Eval}

\begin{figure*}[ht]
	\centering
	\subfigure[Optimal strategy between selfish mining and honest mining.]{\label{SMHM}\includegraphics[width=0.65\columnwidth]{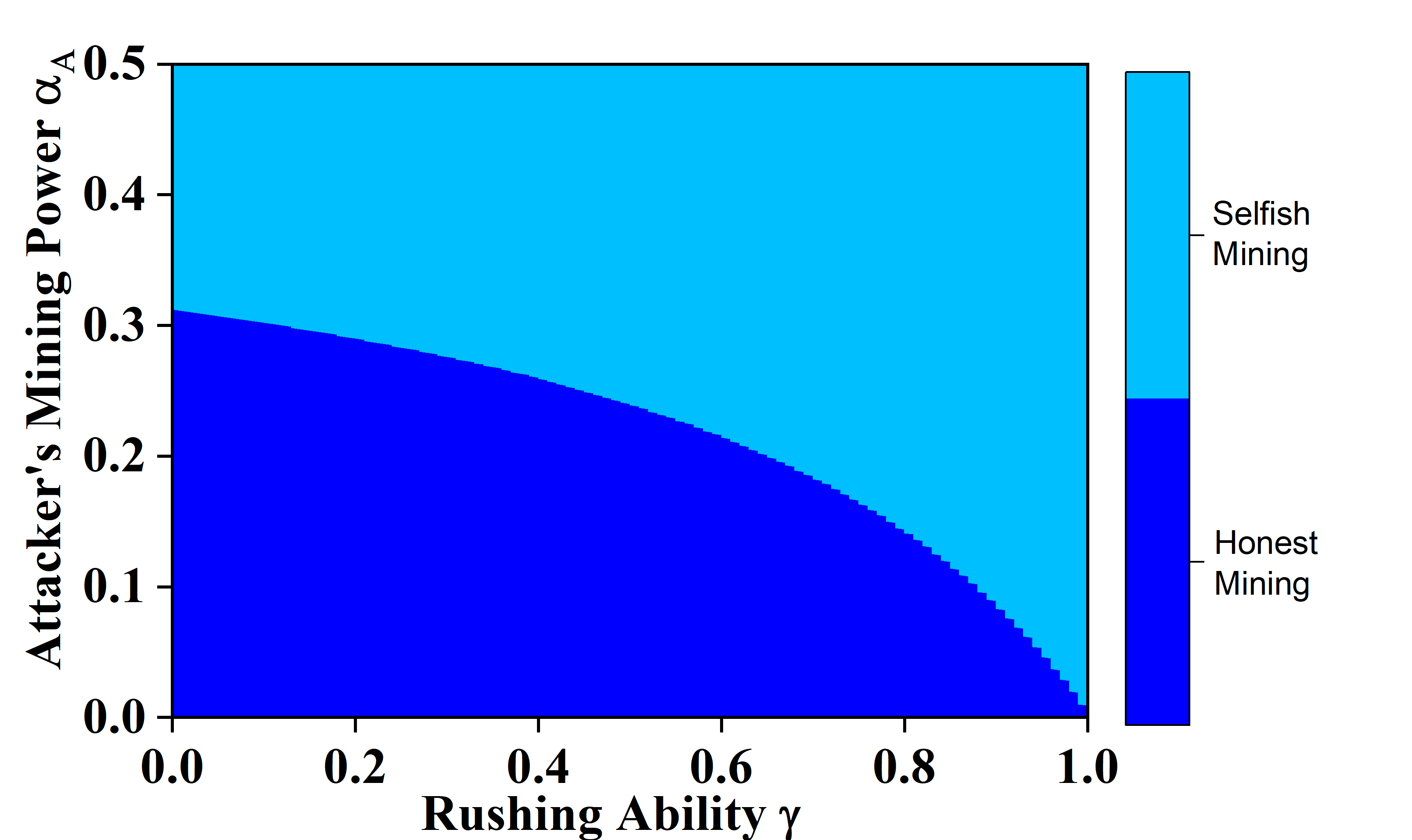}}
	\quad
	\subfigure[Optimal strategy among selfish mining, honest mining, and PSM with different fractions of rational miners.]{\label{PSM}\includegraphics[width=0.65\columnwidth]{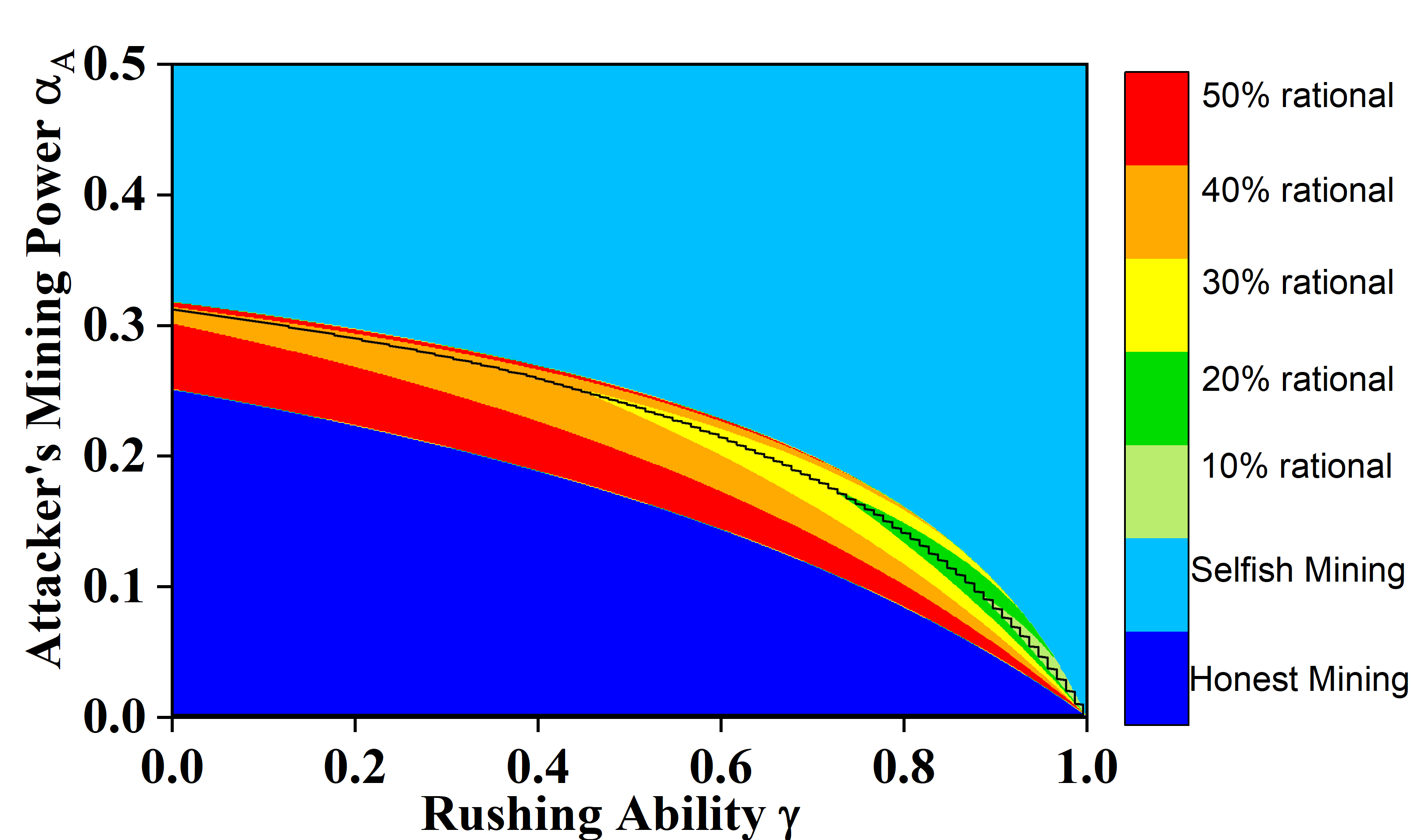}}
	\quad
	\subfigure[The amount of attracted mining power needed for A-PSM attacker to be the optimal strategy.]{\label{AdvancedPSM}\includegraphics[width=0.65\columnwidth]{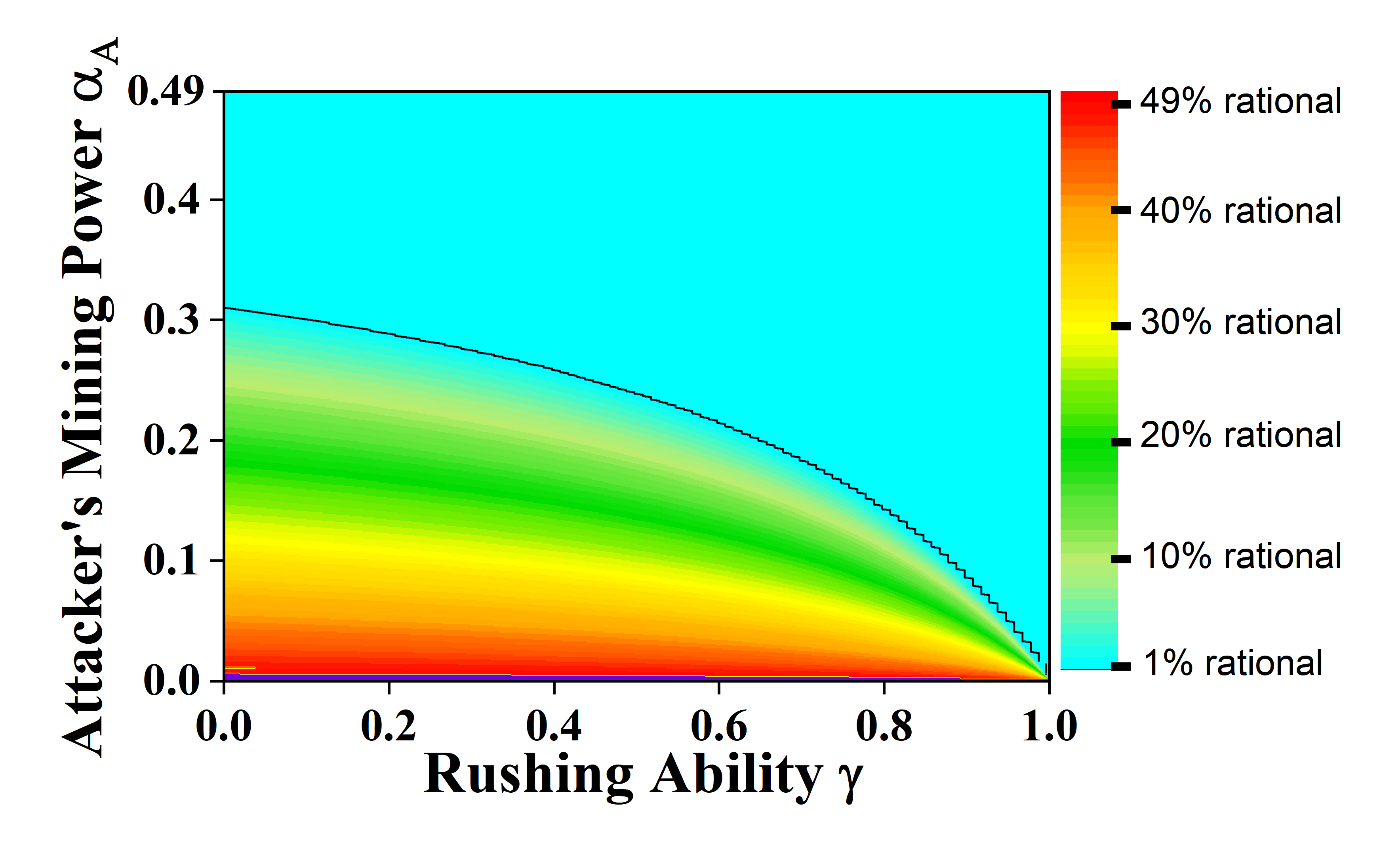}}
	\caption{\textbf{Optimal strategies for different $\alpha_A$ and $\gamma$ values.} We also redraw the borderline between selfish and honest mining in (a), (b), and (c).}
	\quad
\label{fig:partialshare}	
\end{figure*}


In an ideal case with known $\alpha_A$ and $\gamma$, we can get the optimal mining strategy for participants. In this section, we compare PSM, A-PSM with selfish mining and honest mining, and analyze the optimal mining strategy for both rational miners and attackers respectively. \nop{based on Equations (\ref{equ:atkprof}), (\ref{equ:Honest}), (\ref{equ:RevSelfish}), and (\ref{equ:APSMRev}). Note that the A-PSM strategy requires rational miners to trust the attacker more than PSM. Here we make a stricter assumption for the A-PSM strategy, i.e., $\alpha_A+\alpha_i\leq0.5$.}

\subsection{Optimal Mining Strategy for Rational Miners}

Rational miners are profit-driven. They will receive the shared partial block information no matter in PSM or A-PSM. 

In PSM, rational miner's optimal strategy can be decided by Equation (\ref{equ:PSMMinerRER}). When $RER_k^{G,H}$ value is positive, rational miner should follow the attacker's branch. Otherwise, it should follow the public mining branch. In reality, it is hard to obtain the optimal strategy for rational miners because Equation (\ref{equ:PSMMinerRER}) contains the parameter $\gamma$. $\gamma$ is unpredictable that represents the ratio of public miners choosing attacker's branch when a race occurs. Thus, we derive the optimal strategy in Theorem 4.4 for the worst-case scenario when $\gamma=0$. Figure~\ref{fig:RERMinerQ} illustrates the optimal strategy for rational miners.

In A-PSM, rational miner's optimal strategy is decided by Equation (\ref{equ:APSMMinerRER}). Since the formula is always $>0$ in Equation (\ref{equ:APSMMinerRER}), it is more profitable for rational miners to work on the attacker's branch.

\subsection{Optimal Mining Strategy for Attacker}

Figure~\ref{fig:partialshare} shows the optimal mining strategy simulation results for attackers under different conditions. Specifically, Figure~\ref{SMHM} compares selfish mining to honest mining in different network settings, revealing that a selfish miner with over $1/3$ mining power consistently achieves higher profits. 

In PSM, an attacker can get its optimal strategy from Equations (\ref{equ:RevHonest}) and (\ref{equ:RevSelfish}). When both $RER_A$ values are positive, the attacker's optimal strategy is to launch PSM attacks. Otherwise, it should take either honest mining or selfish mining. The parameters $\alpha_A$ and $\alpha_i$ define the distributions of rational and honest miners in equations. In Figure~\ref{PSM}, we show in what conditions PSM can outperform both honest and selfish mining. It shows that even if $\aA+\ai<0.5$, the attacker can still achieve higher profits than both Selfish and Honest mining strategies. For example, when the attacker's mining power $\alpha_A=0.09$, and $\gamma=0.9$, the attacker with rational miner controls only 10\% of mining power can be the optimal mining strategy. When attracted miners control 50\% of mining power, an attacker with 25\% of mining power can outperform selfish and honest mining even if $\gamma =0$.
 
In A-PSM, an attacker can get its optimal strategy from Equations (\ref{equ:APSMHonAtk}) and (\ref{equ:APSMSelfAtk}) similarly. In Figure~\ref{AdvancedPSM}, we graphically illustrate the optimal strategy for attackers in A-PSM. As we can see from Figure~\ref{AdvancedPSM}, A-PSM can outperform selfish mining even if the attracted miners only control less than 1\% of mining power. The fraction of rational miners' mining power needed increases with the decrement of the attacker's mining power. When the attacker's mining power is less than 1\%, the amount of rational miner's mining power needed approaches 49\%. But the overall mining power of attracted miners and the attacker is always no more than 50\%. Whether to take PSM or A-PSM depends on the comparison result between Equation (\ref{equ:atkprof}) \nop{(PSM)} and Equation (\ref{equ:RevSelfish}) \nop{(A-PSM)}.



In reality, $\alpha_A$, $\alpha_i$, and $\gamma$ could change dynamically. The attacker can use network measurement~\cite{saad2021syncattack,saad2021revisiting} or historical information, e.g., its mining power ratio in the past 1 hour, to infer the current network status and select the attack strategy. 



\section{Discussion}
\label{sec:Discussion}

\textbf{Multiple Attackers:} Except for honest miners, a miner who finds a new block can launch PSM or A-PSM attacks simultaneously. Thus, a PoW-based blockchain system can have multiple attackers. The evaluation results indicate that in PSM, the variance of mining power among different attackers has no impact on rational miners; joining either one of them ensures profitability. For the A-PSM strategy, attracted miners must re-evaluate their profits and align with the attacker possessing higher mining power. The detailed proofs for these conclusions are provided in Appendix~\ref{sec:MultiAtk}. 

Additionally, the discussion regarding \textbf{attack mitigation} and the \textbf{rationality of miners} can be found in Appendix~\ref{sec:BlockSharing}. 

\section{Conclusions}
\label{sec:Conclusions}
Selfish mining attack poses constant threats to blockchain systems, especially small-scale public chains. Previous attacking methods do not improve selfish mining itself, but combine it with other techniques, such as bribery attacks, BWH and FAW, to gain more advantages. Unlike previous approaches, we exploit partial information publishing in the mining process, which shows a new mining attack paradigm of colluding with other rational miners for an attacker in this paper.  We first propose Partial Selfish Mining (PSM) attacks. Based on the idea of partial block sharing, PSM attacks allow both attackers and attracted miners to earn more rewards. PSM is also feasible to launch because it has two mechanisms to guarantee attracted miners' profit such that they will follow the attacker's private branch. The proposed Advanced PSM (A-PSM) can improve the attacker's profit to be no less than selfish mining, making it a profitable alternative to current selfish mining in mining-related attacks to PoW blockchain systems. To mitigate partial selfish mining, we discuss some possible countermeasures. However, a practical solution remains to be open.

\bibliographystyle{IEEEtranS}
\bibliography{ccs-sample}
\appendix

\section{Proof of Attacker and Greedy Miners' Profits in PSM}\label{sec:PSMAtk}

\begin{figure}[ht]
	\centering
	\subfigure[State machine of PSM when miners with overall $\alpha_i$ mining power working as attracted miner. ]{\label{fig:PSMStateMachine}\includegraphics[width=0.4\columnwidth]{fig/PSMAttacker.png}}
	\subfigure[State machine of PSM when the only rational miner k with mining power $\ak$ choose to follow public mining strategy.]{\label{fig:PSMMinerStateMachine}\includegraphics[width=0.45\columnwidth]{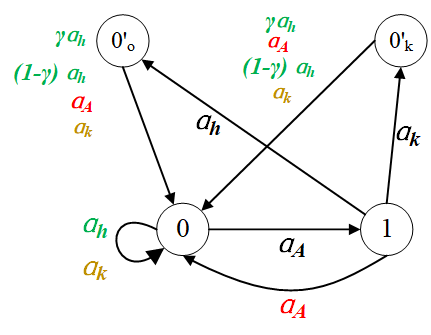}}
   	\caption{State machine of PSM in different conditions.}
\end{figure}


From the state machine of PSM in \fig~\ref{fig:PSMStateMachine}, we can get the following Equations: 
\begin{small}
\begin{subequations}
\begin{align}
\alpha_A P_0 &=P_1,\label{eq:atk1}  \\
P_0'&=\alpha_h P_1,\label{eq:atk2} \\
\alpha_A P_0&=P_0'+(1-\alpha_h)P_1, \label{eq:atk3} \\
P_0&+P_0'+P_1 =1\label{eq:atk4}
\end{align}
\end{subequations}
\end{small}
From Equation~(\ref{eq:atk1})  we can get:
\begin{small}
\begin{equation}
    P_0=\frac{1}{\aA}P_1 \label{eq:atk5}
\end{equation}
\end{small}
Plugging the Equation~(\ref{eq:atk5}) and (\ref{eq:atk2}) into~(\ref{eq:atk4}), then we get
\begin{small}
\begin{equation}
\begin{aligned}
\label{equ:atkP1}
    P_1+\frac{1}{\aA} P_1 +\ah P_1 &=1\\
    P_1 =\frac{\aA}{1+\aA+\aA\ah}.& 
    \end{aligned}
\end{equation}
\end{small}
 With  Equation (\ref{equ:atkP1}), we can get other probabilities:
 \begin{small}
 \begin{subequations}
\begin{align}
P_0&=\frac{1}{1+\aA+\aA\ah},\label{eq:atkP0}  \\
P_0'&=\frac{\aA\ah}{1+\aA+\aA\ah},\label{eq:atkP0'}
\end{align}
\end{subequations}
\end{small}
Then we can calculate the revenue of the attacker as: 
\begin{small}
\begin{equation}
\begin{aligned}
\label{equ:PSMatkrev}
   Rev^{PSM}_{atk}&=2\alpha_A P_1+(2 \alpha_A + \alpha_i + \gamma \alpha_h)P_0'+\alpha_i P_1\\
   Rev^{PSM}_{atk}&=(2\alpha_A+\alpha_i)(\frac{\alpha_A\alpha_h+\alpha_A}{1+\alpha_A+\alpha_A\alpha_h})+\gamma\alpha_h (\frac{\alpha_A\alpha_h}{1+\alpha_A+\alpha_A\alpha_h})\\
   Rev^{PSM}_{atk}&=\frac{(2\aA+\ai)(\aA\ah+\aA)+\ar\aA\ah^2}{1+\aA+\aA\ah}.
    \end{aligned}
\end{equation}
\end{small}
The overall revenue of all the attracted miners is:
\begin{equation}
\begin{aligned}
\label{equ:PSMgreedyrev}
   Rev^{PSM}_{attracted}&= \ai.
    \end{aligned}
\end{equation}

And the overall revenue of all the honest miners are:
\begin{equation}
\begin{aligned}
\label{equ:PSMhonrev}
   Rev^{PSM}_{honest}&= (\ar\ah+2(1-\ar)\ah)P_0'+\ah P_0\\
   Rev^{PSM}_{honest}&= \frac{2\aA\ah^2+\ah-\aA\ah^2\ar}{1+\aA+\aA\ah}.
    \end{aligned}
\end{equation}

Since honest miners or the attacker may work on the blocks that eventually end up outside the blockchain, the overall block generation rate $Rev^{PSM}_{atk}+Rev^{PSM}_{attracted}+Rev^{PSM}_{honest}<1$. Thus, we can calculate the attacker's expected profit as: 
\begin{equation}
\begin{aligned}
\label{equ:PSMatkprof}
   R^A_P=&\frac{Rev^{PSM}_{atk}}{Rev^{PSM}_{atk}+Rev^{PSM}_{attracted}+Rev^{PSM}_{honest}}\\
   =&\frac{\aA\ah^2\ar+(\aA\ah+\aA)\ai+2\aA^2\ah+2\aA^2}{(2\aA\ah+2\aA+1)\ai+2\aA\ah^2+(2\aA^2+1)\ah+2\aA^2}
    \end{aligned}
\end{equation}

And the attracted miners' overall expected profit is:
\begin{equation}
\begin{aligned}
\label{equ:PSMratprof}
   R^G_P=&\frac{Rev^{PSM}_{attracted}}{Rev^{PSM}_{atk}+Rev^{PSM}_{attracted}+Rev^{PSM}_{honest}}\\
   =&\frac{(\aA\ah+\aA+1)\ai}{(2\aA\ah+2\aA+1)\ai+2\aA\ah^2+(2\aA^2+1)\ah+2\aA^2}
    \end{aligned}
\end{equation}

\section{Proof of Rational Miner's Profits in PSM}\label{sec:PSMMiner}
From the state machine in \fig~\ref{fig:PSMMinerStateMachine}, we can get the following Equations: 
\begin{subequations}
\begin{align}
\alpha_A P_0 &=P_1,\label{eq:matk1}  \\
P_{0o}'&=\alpha_h P_1,\label{eq:matk2} \\
P_{0k}'&=\alpha_k P_1,\label{eq:matk3} \\
P_0+&P_{0o}'+P_{0k}'+P_1 =1\label{eq:matk4}
\end{align}
\end{subequations}

From Equation~(\ref{eq:matk1})  we can get:
\begin{equation}
    P_0=\frac{1}{\aA}P_1 \label{eq:matk5}
\end{equation}
Plugging the Equation~(\ref{eq:matk5}) (\ref{eq:matk2}) and (\ref{eq:matk3}) into~(\ref{eq:matk4}), then we get
\begin{equation}
\begin{aligned}
\label{equ:matkP1}
    P_1+\frac{1}{\aA} P_1 +\ah P_1 +\ak P_1 &=1\\
    P_1 =\frac{\aA}{1+\aA(1+\ak+\ah)}.& 
    \end{aligned}
\end{equation}

 With Equation~\ref{equ:matkP1}, we can get other probabilities:
 \begin{subequations}
\begin{align}
P_0&=\frac{1}{1+\aA(1+\ak+\ah)},\label{eq:matkP0}  \\
P_{0o}'&=\frac{\aA\ah}{1+\aA(1+\ak+\ah)},\label{eq:matkP0'h}\\
P_{0k}'&=\frac{\aA\ak}{1+\aA(1+\ak+\ah)},\label{eq:matkP0'k}.
\end{align}
\end{subequations}

With $\ah=1-\aA-\ak$, we can calculate the revenue of the miner k as: 

\begin{equation}
\begin{aligned}
\label{equ:PSMkrev}
   Rev^{PSM,k}_{k}&=(2\alpha_k+(1-\ar)\ah) P_{0k}'+\alpha_k(P_0+P_{0o}')\\
   Rev^{PSM,k}_{k}&=\frac{2\aA\ak^2+(1-\ar)\ah\aA\ak}{1+\aA(1+\ak+\ah)}+\frac{\aA\ak(\ak+\ah)}{1+\aA(1+\ak+\ah)}\\
Rev^{PSM,k}_{k}&=\frac{((\aA-\aA^2)\ak-\aA\ak^2)r+(2\aA^2-2\aA-1)\ak}{\aA^2-2\aA-1}.
    \end{aligned}
\end{equation}

The overall revenue of all other honest miners is:
\begin{equation}
\begin{aligned}
\label{equ:PSMorev}
   Rev^{PSM,k}_{o}=&\ah(P_0+P_{0k}')+(\ah\ar+2\ah(1-\ar))P_{0o}'+\\
   &(1-\ar)\ak P_{0o}'\\
   Rev^{PSM,k}_{o}=&\frac{((\aA^2-\aA)\ak+\aA^3-2\aA^2+\aA)\ar}{\aA^2-2\aA-1}+\\
   &\frac{((-2\aA^2)+2\aA+1)\ak}{\aA^2-2\aA-1}-\frac{2\aA^3-4\aA^2+\aA+1}{\aA^2-2\aA-1}.
    \end{aligned}
\end{equation}

And the revenue of the attacker is:
\begin{equation}
\begin{aligned}
\label{equ:PSMmineratkrev}
   Rev^{PSM,k}_{atk}&=2\aA (P_{0o}'+P_{0k}'+P_1)+\ar\ah P_{0k}'+\ar (\ah +\ak) P_{0o}'\\
   Rev^{PSM,k}_{atk}&=\frac{(\aA\ak^2-\aA^3+2\aA^2-\aA)\ar+2\aA^3-4\aA^2}{\aA^2-2\aA-1}.
    \end{aligned}
\end{equation}

And the miner k's overall expected profit is:
\begin{equation}
\begin{aligned}
\label{equ:kprof}
   R^k_P=&\frac{Rev^{PSM,k}_{k}}{Rev^{PSM,k}_{atk}+Rev_{k}+Rev_{o}}\\
   =&\frac{\aA\ah^2\ar+(\aA\ah+\aA)\ak+2\aA^2\ah+2\aA^2}{2\ah^2\ar+(2\aA\ah+2\aA+1)\ak+\aA\ah^2+2\aA^2\ah+2\aA^2}\\
   =&\frac{(\aA\ak^2+(\aA^2-\aA)\ak)\ar+((-2\aA^2)+2\aA+1)\ak}{\aA+1}
    \end{aligned}
\end{equation}

The relative extra reward of miner k when following attracted mining instead of public mining strategy is: 
\begin{equation}
\begin{aligned}
   RER_k^{G,H} &=\frac{R_k^G-R_k^H}{R_k^H}\\
   &=\frac{((-\aA\ai)-\aA^2+\aA)\ar-\aA\ai+\aA^2}{(\aA\ai+\aA^2-\aA)\ar-2\aA^2+2\aA+1}
    \end{aligned}
\end{equation}

\section{Proof of Attacker and attracted miners' Profits in A-PSM}\label{sec:APSMattacker}


\begin{figure}[ht]
	\centering
	\subfigure[State machine of A-PSM when miners with overall $\alpha_k$ mining power working as attracted miner. ]{\label{fig:APSMStateMachine}\includegraphics[width=\columnwidth]{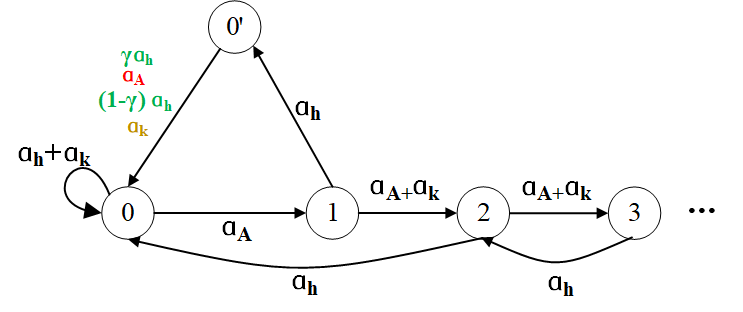}}
	\subfigure[State machine of A-PSM when miner k choose to follow public mining strategy. ]{\label{fig:APSMMinerStateMachine}\includegraphics[width=\columnwidth]{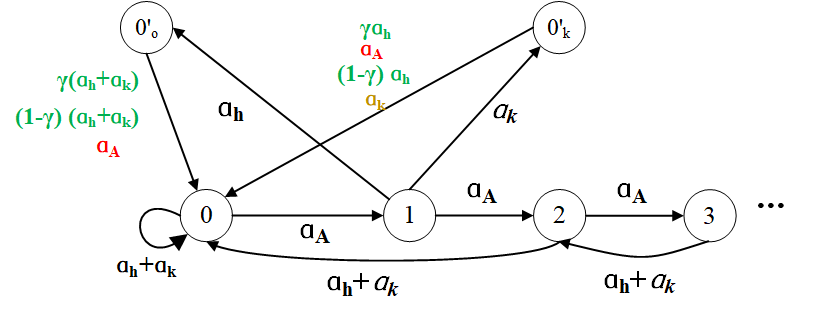}}
   	\caption{State machine of PSM in different conditions.}
\end{figure}

From the state machine in  \fig~\ref{fig:APSMStateMachine}, we can get the following equations: 

\begin{subequations}
\begin{align}
\alpha_A P_0 &=P_1,\label{eq:apsmatk1}  \\
P_0'&=\alpha_h P_1,\label{eq:apsmatk2} \\
(1-\ah)P_n-1&=\ah P_{n}, \forall n\geq 2  \label{eq:apsmatk3} \\
P_0'+\sum_{k=0}^\infty P_k &=1\label{eq:apsmatk4}
\end{align}
\end{subequations}

From equation~(\ref{eq:apsmatk2})  we can get:
\begin{equation}
    P_0'=\aA\ah P_0 \label{eq:apsmatkp0d}
\end{equation}
Then, equation~(\ref{eq:apsmatk3}) and (\ref{eq:apsmatk1}) give us:
\begin{equation}
\begin{aligned}
\label{equ:apsmatkPn}
P_n =\aA(\frac{1-\ah}{\ah})^{n-1}P_0, \forall n\geq 2
    \end{aligned}
\end{equation}

Then we can express equation (\ref{eq:apsmatk4}) as
 \begin{subequations}
\begin{align}
P_0+\aA(1+\ah)P_0+\frac{\aA(1-\ah)}{2\ah-1}P_0=1
\label{eq:apsmp0}
\end{align}
\end{subequations}
From equation (\ref{eq:apsmp0}) we can obtain $P_0$ and other probabilities:
 \begin{subequations}
\begin{align}
P_0&=\frac{2\ah-1}{2\aA\ah^2+2\ah-1}\\
P_0'&=\frac{\aA\ah(2\ah-1)}{2\aA\ah^2+2\ah-1}\\
P_1&=\frac{\aA(2\ah-1)}{2\aA\ah^2+2\ah-1}\\
P_n&=(\frac{1-\ah}{\ah})^{n-1}\frac{\aA(2\ah-1)}{2\aA\ah^2+2\ah-1},\forall n\geq 2.
\label{eq:apsmresult}
\end{align}
\end{subequations}

Then we can calculate the revenue of the attacker as: 
\begin{small}
\begin{equation}
\begin{aligned}
\label{equ:APSMatkrev}
   Rev^{APSM}_{atk}&=2\aA P_0'+\ai P_0'+r\ah P_0'+(1+\frac{\aA}{\aA+\ai})\ah P_2\\
   &+(\frac{\aA}{\aA+\ai})\ah \sum_{q=3}^\infty P_q\\
   Rev^{APSM}_{atk}&=\frac{(-\aA(\ai+\aA-1)^2(2\ai+2\aA-1))\ar}{2\aA\ai^2+(4\aA^2-4\aA-2)\ai+2\aA^3-4\aA^2+1}+\\
   &\frac{2\aA\ai^3+(8\aA^2-5\aA)\ai^2+(10\aA^3-14\aA^2+2\aA)\ai}{2\aA\ai^2+(4\aA^2-4\aA-2)\ai+2\aA^3-4\aA^2+1}+\\
   &\frac{4\aA^4-9\aA^3+4\aA^2}{2\aA\ai^2+(4\aA^2-4\aA-2)\ai+2\aA^3-4\aA^2+1}.
    \end{aligned}
\end{equation}
\end{small}
The overall revenue of all the attracted miners is:
\begin{equation}
\begin{aligned}
\label{equ:APSMgreedyrev}
   Rev^{APSM}_{attracted}&= \ai.
    \end{aligned}
\end{equation}

And the overall revenue of all the honest miners are:
\begin{equation}
\begin{aligned}
\label{equ:APSMhonrev}
   Rev^{APSM}_{honest}&= \ah P_0+2\ah(1-\ar)P_0'+\ar\ah P_0'\\
   Rev^{APSM}_{honest}&= \frac{\aA(\ai+\aA-1)^2(2\ai+2\aA-1)(\ar-3)}{2\aA\ai^2+(4\aA^2-4\aA-2)\ai+2\aA^3-4\aA^2+1}.
    \end{aligned}
\end{equation}

Since honest miners or the attacker may work on the blocks that eventually end up outside the blockchain, the overall block generation rate $Rev_{atk}+Rev_{attracted}+Rev_{honest}<1$. Thus, we can calculate the attacker's expected profit as: 
\begin{equation}
\begin{aligned}
\label{equ:APSMatkprofit}
   R^A_P=&\frac{Rev^{APSM}_{atk}}{Rev^{APSM}_{atk}+Rev^{APSM}_{attracted}+Rev^{APSM}_{honest}}\\
   =&\frac{(-\aA(\ai+\aA-1)^2(2\ai+2\aA-1))\ar}{\aA\ai^2+(2\aA^2-2\aA-2)\ai+\aA^3-2\aA^2-\aA+1}+\\
&\frac{2\aA\ai^3+(8\aA^2-5\aA)\ai^2+(10\aA^3-14\aA^2+2\aA)\ai}{\aA\ai^2+(2\aA^2-2\aA-2)\ai+\aA^3-2\aA^2-\aA+1}+\\
&\frac{4\aA^4-9\aA^3+4\aA^2}{\aA\ai^2+(2\aA^2-2\aA-2)\ai+\aA^3-2\aA^2-\aA+1}
    \end{aligned}
\end{equation}

And the rational miners' overall expected profit is:
\begin{equation}
\begin{aligned}
\label{equ:APSMratprof}
   R^G_P=&\frac{Rev^{APSM}_{attracted}}{Rev^{APSM}_{atk}+Rev^{APSM}_{attracted}+Rev^{APSM}_{honest}}\\
   =&\frac{2\aA\ai^3+(4\aA^2-4\aA-2)\ai^2+(2\aA^3-4\aA^2+1)\ai}{\aA\ai^2+(2\aA^2-2\aA-2)\ai+\aA^3-2\aA^2-\aA+1}
    \end{aligned}
\end{equation}

\section{Proof of One Rational Miner's Profits in APSM}\label{sec:APSMMiner}
From the state machine in~\fig~\ref{fig:APSMMinerStateMachine}, we can get the following equations: 
\begin{subequations}
\begin{align}
P_1&=\aA P_0,\label{eq:apsmmatk1}  \\
P_{0o}'&=\ah P_1,\label{eq:apsmmatk2} \\
P_{0k}'&=\alpha_k P_1,\label{eq:apsmmatk3} \\
P_q&=\frac{\aA}{1-\aA}P_{q-1}, \forall q\geq 2\label{eq:apsmmatk4}\\
P_0&+P_{0o}'+P_{0k}'+\sum_{q=1}^\infty =1\label{eq:apsmmatk5}
\end{align}
\end{subequations}

Combining the equation~(\ref{eq:apsmmatk1}) with (\ref{eq:apsmmatk2}) and (\ref{eq:apsmmatk3}), then we get
\begin{equation}
\begin{aligned}
P_{0k}'&=\aA\ak P_0\\
P_{0o}'&=\aA\ah P_0.
    \end{aligned}
\end{equation}

With equation (\ref{eq:apsmmatk4}), we can get: 
\begin{equation}
\begin{aligned}
P_n=(\frac{\aA}{1-\aA})^{n-1}\aA P_0
    \end{aligned}
\end{equation}

 With equation~\ref{eq:apsmmatk5}, we can get all states probabilities:
 \begin{subequations}
\begin{align}
P_0&=\frac{1-2\aA}{2\aA^3-4\aA^2+1},\label{eq:APSMmatkP0}  \\
P_{0o}'&=\frac{\aA\ah(1-2\aA)}{2\aA^3-4\aA^2+1},\label{eq:APSMmatkP0'h}\\
P_{0k}'&=\frac{\aA\ak(1-2\aA)}{2\aA^3-4\aA^2+1},\label{eq:APSMmatkP0'k}\\
P_n &= (\frac{\aA}{1-\aA})^{n-1}\aA \frac{1-2\aA}{2\aA^3-4\aA^2+1}
\end{align}
\end{subequations}

Then we can calculate the revenue of the miner k as: 

\begin{equation}
\begin{aligned}
\label{equ:krev}
   Rev^{APSM,k}_{k}&=\ak P_0+2\ak P'_{0k}+\ak P'_{0o}+(1-\ar)\ah P'_{0k}\\
Rev^{APSM,k}_{k}&=\frac{((\aA-2\aA^2)\ak^2+((-2\aA^3)+3\aA^2-\aA)\ak)\ar}{2\aA^3-4\aA^2+1}+\\
&\frac{(4\aA^3-6\aA^2+1)\ak}{2\aA^3-4\aA^2+1}.
    \end{aligned}
\end{equation}

The revenue of other honest miners are:
\begin{small}
\begin{equation}
\begin{aligned}
\label{equ:orev}
   Rev^{APSM,k}_{o}&=\ah(P_0+P_{0k}')+(\ah\ar+2\ah(1-\ar))P_{0o}'+(1-\ar)\ak P_{0o}'\\
   Rev^{APSM,k}_{o}&=\frac{(2\aA-1)(\ai+\aA-1)(\aA^2\ar-\aA\ar-2\aA^2+2\aA+1)}{2\aA^3-4\aA^2+1}.
    \end{aligned}
\end{equation}
\end{small}
And the revenue of the attacker is:
\begin{equation}
\begin{aligned}
\label{equ:atkerrev}
   Rev^{APSM,k}_{atk}=&2\aA(P'_{0k}+P'_{0o})+\ar(1-\aA)P'_{0o}+\ar\ah P'_{0k}+\\
   &2(1-\aA)P_2+(1-\aA) \sum_{q=3}^\infty P_q\\
Rev^{APSM,k}_{atk}=&\frac{(\aA(2\aA-1)(\ak-\aA+1)(\ak+\aA-1))\ar}{2\aA^3-4\aA^2+1}+\\
&\frac{4\aA^4-9\aA^3+4\aA^2}{2\aA^3-4\aA^2+1}.
    \end{aligned}
\end{equation}

And the miner k's overall expected profit is:
\begin{equation}
\begin{aligned}
\label{equ:minerkprof}
   R^k_P=&\frac{Rev^{APSM,k}_{k}}{Rev^{APSM,k}_{atk}+Rev^{APSM,k}_{k}+Rev^{APSM,k}_{o}}\\
   =&\frac{(\aA(1-2\aA)\ak(\ak+\aA-1))\ar+(4\aA^3-6\aA^2+1)\ak}{\aA^3-2\aA^2-\aA+1}
    \end{aligned}
\end{equation}

The relative extra reward of miner k when following attracted mining instead of public mining strategy is: 
\begin{equation}
\begin{aligned}
   RER_k^{G,H} &=\frac{R_k^G-R_k^H}{R_k^H}\\
   &=\frac{((-\aA\ai)-\aA^2+\aA)\ar-\aA\ai+\aA^2}{(\aA\ai+\aA^2-\aA)\ar-2\aA^2+2\aA+1}
    \end{aligned}
\end{equation}

\section{Multiple Attackers in PSM and A-PSM}
\label{sec:MultiAtk}
Here, we consider a model with two attackers. Our two-attacker analysis results can be easily extended to the multiple-attacker scenario because the given relative extra reward of attackers still holds.
Rational miners can choose to work on one of the attacker's private branches or continue working on the public branch. 

First, we consider the case that two attackers release the partial block simultaneously. Assuming there is an attacker $A$ with mining power $\alpha_A$, and rushing ability $\gamma_A$. Attacker $B$ with mining power $\alpha_B$ and rushing ability $\gamma_B$. Rational miner $k$ with mining power $\alpha_k$ considers that the rest of the miners are public miners, which means $\alpha_h=1-\alpha_A-\alpha_B-\alpha_k$. 

For PSM, the relative extra reward of choosing attacker $A$ instead of following the public mining strategy is
\begin{small}
\begin{equation}
  RER_{k,PSM}^{A,H}=  \frac{1-2\alpha_k-(1-\gamma)(1-\alpha_A-\alpha_B-\alpha_k)}{2\alpha_k+(2-\gamma)(1-\alpha_A-\alpha_B-\alpha_k)}
\end{equation}
\end{small}

The RER of mining in attacker $A$ over $B$'s branch is always 0. That means the variance of mining power among different attackers has no impact on rational miners. Joining either one of them could ensure its profits. 
 
 For A-PSM, the RER of choosing attacker $A$ instead of following the public mining strategy is
 
  \begin{small}
 \begin{equation}
 \begin{aligned}
     RER_{k,A-PSM}^{A,H}=\frac{\alpha_h+\frac{1}{\alpha_A+\alpha_i}(\frac{(\alpha_A+\alpha_i)^2}{1-2(\alpha_A+\alpha_i)})+1}{2\alpha_i+(2-\gamma)\alpha_h}-1.
     \end{aligned}
 \end{equation}
 \end{small}
 The RER of choosing attacker $A$ instead of attacker $B$ is
 \begin{small}
 \begin{equation}
 \begin{aligned}
     RER_{k,A-PSM}^{A,B}=\frac{\alpha_A-\alpha_B}{(2\ai+2\alpha_A-1)(2\ai+2\alpha_B-1)}.
     \end{aligned}
 \end{equation}
 \end{small}
 It is always more beneficial for rational miners to work with larger mining power attackers.
 
Then we consider one of the attackers to find the new block first. Assuming attacker $A$ finds the new block at time $t_a$, the attacker $B$ finds the new block at time $t_b$, duration $T_d=t_b-t_a>0$. Since the rational miner trusts both attackers equally, joining the attacker $A$'s private branch means the rational miner can have $\alpha_h R(T_d)\alpha_i+\alpha_iR(T_d)$ more expected revenue during the period $T_d$. Note that the definition of $R(T)$ is given in Equation (\ref{equ:RT}). 
Thus, attracted miners will tend to join the attacker $A$'s private branch first to assure higher profits for both PSM and A-PSM strategies. If the rational miners do not find the new block during the period $T_d$, with the PSM strategy, the attracted miners have no motivation to switch the working branch. For the A-PSM strategy, the attracted miners need to re-evaluate their profits and work with the attacker with higher mining power.

\section{Discussion}
\label{sec:BlockSharing}


\subsection{Block Sharing Strategy}

Before the exchange starts, the attacker can deploy a smart contract as an arbiter. Then, the attacker and the miner achieve a fair exchange of the partial block data through a simple interaction:

\begin{enumerate}
    \item The attacker generates a random key $k$ and encrypts the $b$ value with $k$,       i.e., $\hat{b}=\mathsf{Enc}(b,k)$. Next, it computes the hash of $k$ $h_k\leftarrow H(k)$ and generates a proof:
    \begin{small}
    \begin{equation}
        \pi_e\overset{\rm R}\leftarrow\mathsf{Prove}((\hat{b},h, h_k),(b,k,nonce,r))
    \end{equation}
    \end{small}
    to prove the following relation:
    \begin{small}
    \begin{equation}
    \begin{aligned}
        H(b,nonce,r)=h\ \land\ 0 & \leq {\rm nonce}\leq2^{32}-1\\
        \land\ \hat{b}=\mathsf{Enc}(b,k) & \land h_k=H(k)
    \end{aligned}
    \end{equation}
    \end{small}
    is satisfied. Finally, it sends the tuple $(\pi_e,\hat{b},h,h_k)$ to the miner.
    \item After receiving the tuple $(\pi_e,\hat{b},h,h_k)$ from the attacker, the miner can verify whether $\pi_e$ is valid by computing $b\leftarrow\mathsf{Verify}((\hat{b},h, h_k),\pi_e)$. If $b=1$, the miner deposits to the arbiter contract the payment agreed upon by both parties beforehand, as well as the hash $h_k$.
    \item  The attacker checks whether the $h_k$ provided by the miner is valid and whether the payment deposited is as previously agreed. If so, the attacker sends the key $k$ to the arbiter contract.
    \item The arbiter contract verifies that $h_k=H(k)$ holds. If it is the case, the contract transfers the payment to the attacker, otherwise, it returns the payment to the miner. Since the $k$ disclosed to the contract will also be available to the miner at the same time, the miner can decrypt $\hat{b}$ by $b\leftarrow\mathsf{Dec}(\hat{b},k)$.
\end{enumerate}

Any miner as a buyer cannot obtain any information about $b$ without completing the payment. Meanwhile, any attacker acting as a seller cannot cheat the payment by submitting the wrong $b$. In this way, the attacker is able to sell the partial block data $b$ to a specific miner and gain revenue.

\subsection{Mitigation}
One straightforward way to prevent PSM attacks is to forbid the broadcasting of partial block data in the target blockchain network. As stated, the partial block data do not need to be propagated through the target network. Rather, it can be obtained from a public website. It would not be easy to forbid the acquirement of partial block data.

Public miners can refuse to mine on the attacker’s branch when a race occurs or thoroughly scrutinize the smart contract to thwart the PSM. Both methods can decrease the attacker’s success rate. The attacker cannot tell which miners are rational. Some public miners can also get a partial block and refuse to mine on the attacker’s branch with the hash value from the partial block data. This reduces the $\gamma$ value.

As shown in Section~\ref{sec:PracticalConcerns}, we denote the hidden data as \textit{secret}. Theoretically, other miners could calculate a valid secret and announce it after getting the partial block, especially for those public miners with sizeable computational power.

Assuming that the attacker publishes a partial block at time $t_0$, the public miner calculates the secret at $t_1$. Let $T=t_1-t_0$. The expected profits of attracted miners can be expressed as $\alpha_h R(T)\alpha_i+\alpha_i R(T)$.
When $T\rightarrow 0$, the attracted miners' expected extra revenue approaches 0.
\nop{ \begin{figure}[htbp]
	\centerline{\includegraphics[width=0.6\columnwidth]{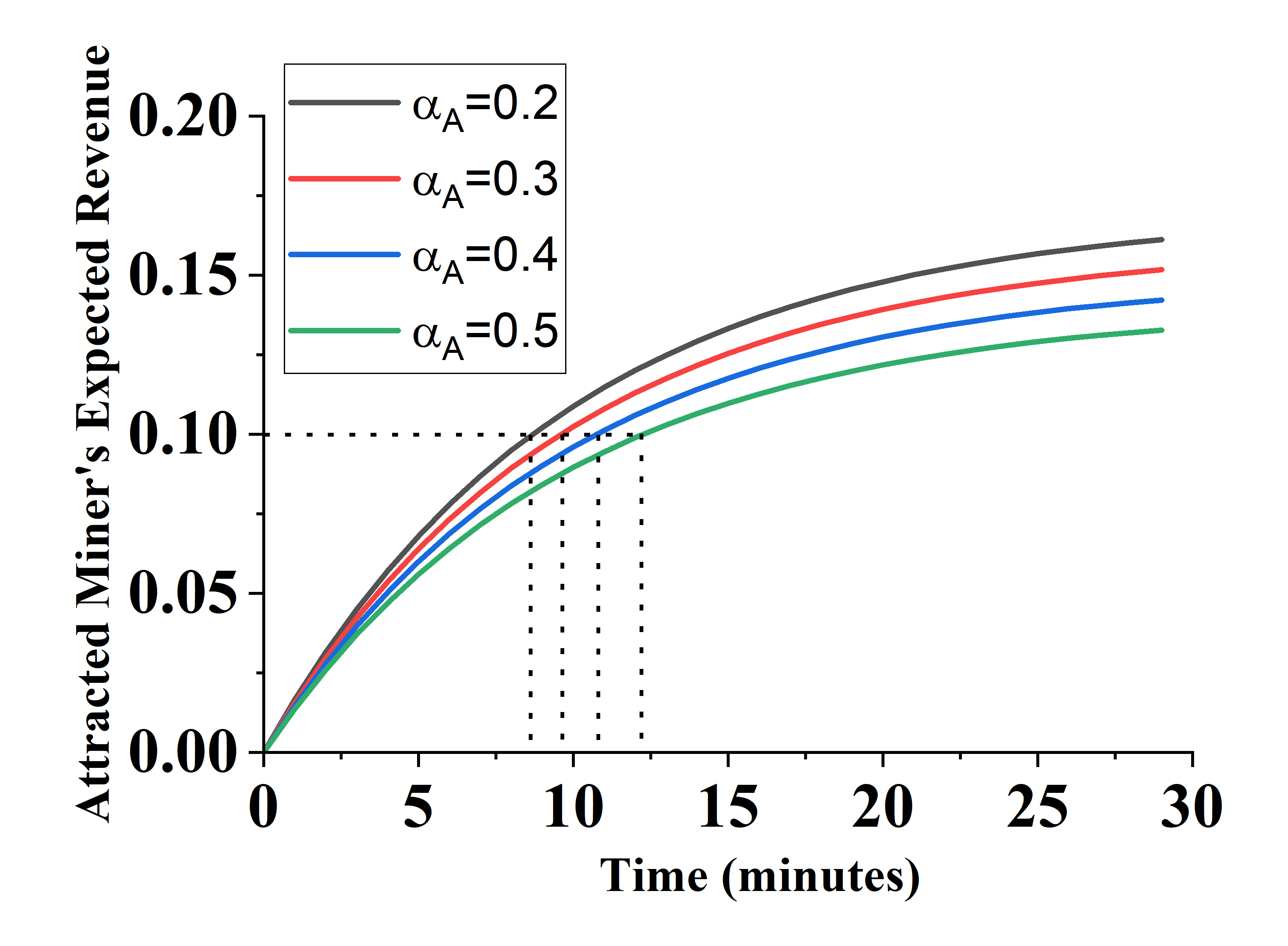}}
	\caption{\textbf{The attracted miners expected revenue with different mining duration.}}
	\label{fig:TimeImpact}
\end{figure}
}

If a public miner releases the secret, then none of the attracted miners will take the risk of getting negligible profits. When $\alpha_i=0$, for the attacker following the A-PSM strategy, its revenue downgrades to selfish mining. The revenue following the PSM strategy is less than selfish mining (see Section \ref{sec:SelfMin} for detail).

Another possible defense against PSM and A-PSM attacks is to reduce the efficiency of partial block dissemination. Public miners can drop the partial block when receiving it. They may put neighbors into an untrust neighbor list who have sent them partial blocks to discourage the dissemination.

\subsection{Rationality of Miners}
In this paper, we assume that no more than 50\% miners are rational. This is a strict and valid assumption.

First, in spite of Nakamoto's claim that block rewards motivate miners to work honestly, miners' altruism is never guaranteed~\cite{sliwinski2021blockchains}. Evidence shows that miners' behavior is profit-driven rather than honest~\cite{eghbali201912}. As previous studies have shown, attackers can exploit miners' profit-driven behavior and sabotage blockchains~\cite{mirkin2020bdos}. Thus, altruism cannot ensure blockchain security. 

Second, in terms of blockchain security analysis, we need to discuss the worst-case but not optimistic ones. The higher ratio of rational miners, the higher profits for attackers. Thus, we believe it is reasonable to consider a strict assumption with no more than 50\% of miners to be rational. 

Third, as shown in \fig~\ref{PSM}, PSM attackers can still get more rewards than both selfish and honest mining, even with a small fraction of rational miners.

\end{document}